\documentclass[10pt, journal, letterpaper]{IEEEtran}
\usepackage{amsmath,amssymb,amsfonts}
\usepackage{amsthm}
\usepackage{bm}
\usepackage{dsfont}
\usepackage{algorithm,algorithmicx,enumitem,algpseudocode}
\usepackage{graphicx}
\usepackage{xcolor}
\usepackage{caption}
\usepackage{subcaption}

\newcommand{\algA}{$\mathrm{BF}$}
\newcommand{\algB}{$\mathrm{VQS}$}
\newtheorem{theorem}{Theorem}
\newtheorem{theorem*}{Subtheorem}
\newtheorem{lemma*}{Lemma}
\newtheorem{corollary}{Corollary}
\newtheorem{proposition}{Proposition}
\newtheorem{definition}{Definition}

\DeclareMathOperator*{\argmax}{arg\,max}

\theoremstyle{remark}
\newtheorem{remark}{Remark}

\usepackage{breqn}
\usepackage{enumitem}
\usepackage[normalem]{ulem}

\newcommand{\be}{\begin{eqnarray}}

\newcommand{\ee}{\end{eqnarray}}
\newcommand{\ben}{\begin{eqnarray*}}
	\newcommand{\een}{\end{eqnarray*}}
\newcommand{\bfl}{\begin{flalign*}}
	\newcommand{\efl}{\end{flalign*}}
\newcommand{\dref}[1]{(\ref{#1})}

\newcommand{\prob}[1]{{\mathbb P} \left( #1\right)}

\newcommand{\calQ}{{\mathcal Q}}
\newcommand{\calD}{{\mathcal D}}
\newcommand{\calA}{{\mathcal A}}

\newcommand{\calL}{{\mathcal L}}

\newcommand{\calH}{{\mathcal H}}

\newcommand{\mVQ}{\mathrm{VQ}}
\newcommand{\state}[1]{{#1}}

\newenvironment{subproof}[1][\proofname]{%
	\begin{proof}[#1]%
	}{%
	\end{proof}%
}

\begin{document}

\title{Scheduling Jobs with Random Resource Requirements in Computing Clusters
\thanks{The authors are with the Department of Electrical Engineering at Columbia University, New York, NY 10027, USA. Emails: \{kpsychas, jghaderi\}@ee.columbia.edu.
	This research was supported in part by NSF Grants CNS-1652115 and CNS-1717867.
}
}
\author{
\IEEEauthorblockN{Konstantinos Psychas, Javad Ghaderi}\\
\IEEEauthorblockA{
Electrical Engineering Department, Columbia University
		} 
}

\maketitle
\begin{abstract}
We consider a natural scheduling problem which arises in many distributed computing frameworks. 
Jobs with diverse resource requirements (e.g. memory requirements) arrive over time 
and must be served by a cluster of servers, each with a finite resource capacity. To improve 
throughput and delay, the scheduler can pack as many jobs as possible in the servers 
subject to their capacity constraints.
Motivated by the ever-increasing complexity of workloads in shared clusters, 
we consider a setting where the jobs' resource requirements belong to a very large number of diverse types or, in the extreme, even infinitely many types, e.g. when resource requirements are drawn from an \textit{unknown} distribution over a continuous support. 
The application of classical scheduling approaches that crucially 
rely on a predefined finite set of types is discouraging in this high (or infinite) dimensional setting. 
We first characterize a fundamental limit on the maximum throughput in such setting,
and then develop oblivious scheduling algorithms that have \textit{low complexity} and can achieve \textit{at least} 1/2 and 2/3 of the maximum throughput, \textit{without the knowledge of traffic or resource 
requirement distribution}. Extensive simulation results, using both synthetic and real traffic 
traces, are presented to verify the performance of our algorithms.

\end{abstract}


\section{Introduction}

Distributed computing frameworks (e.g., MapReduce~\cite{hadoop}, 
Spark~\cite{spark}, Hive~\cite{hive}) have enabled processing of 
very large data sets across a cluster of servers. The processing is 
typically done by executing a set of jobs or tasks in the servers.
A key component of such systems is the resource manager (\textit{scheduler}) 
that assigns incoming jobs to servers and reserves the requested resources 
(e.g. CPU, memory) on the servers for running jobs. For example, in Hadoop~\cite{hadoop}, 
the resource manager reserves the requested resources, by launching \textit{resource containers} in servers. 
Jobs of various applications can arrive to the cluster, which often 
have very diverse resource requirements. Hence, to improve throughput and delay, a scheduler 
should pack as many jobs (containers) as possible in the servers, while retaining their 
resource requirements and not exceeding server's capacities.  

A salient feature of resource demand is that it is hard to predict and 
cannot be easily classified into a small or moderate number of resource profiles 
or ``\textit{types}''. This is amplified by the increasing complexity of workloads, 
i.e., from traditional batch jobs, to queries, graph processing, streaming, 
machine learning jobs, etc., that rely on multiple computation frameworks, and 
all need to share \textit{the same cluster}. For example, Figure~\ref{fig:simgtracestats} 
shows the statistics of memory and CPU resource requirement requested by jobs in a Google 
cluster~\cite{ClusterData}, over the first day in the trace. If jobs were to be divided into types according to their memory requirement alone, there would be more than $700$ types. Moreover, the statistics change over time and these types are not sufficient to model all the job requirements in a month, which are more 
than $1500$. We can make a similar observation for CPU requirements,
which take more than $400$ discrete types. Analyzing the joint CPU and 
memory requirements, there would be more than $10,000$ distinct types.
Building a low-complexity scheduler that can provide high performance in such a high-dimensional 
regime is extremely challenging, as learning the demand for all types is 
infeasible, and finding the optimal packing of jobs in servers, even 
when the demand is known, is a hard combinatorial problem (related to \textit{Bin 
Packing} and \textit{Knapsack} problems~\cite{Martello1990}).

\begin{figure}
	\centering
    \begin{subfigure}{.5\columnwidth}
  	\centering
	\includegraphics[width=\columnwidth]{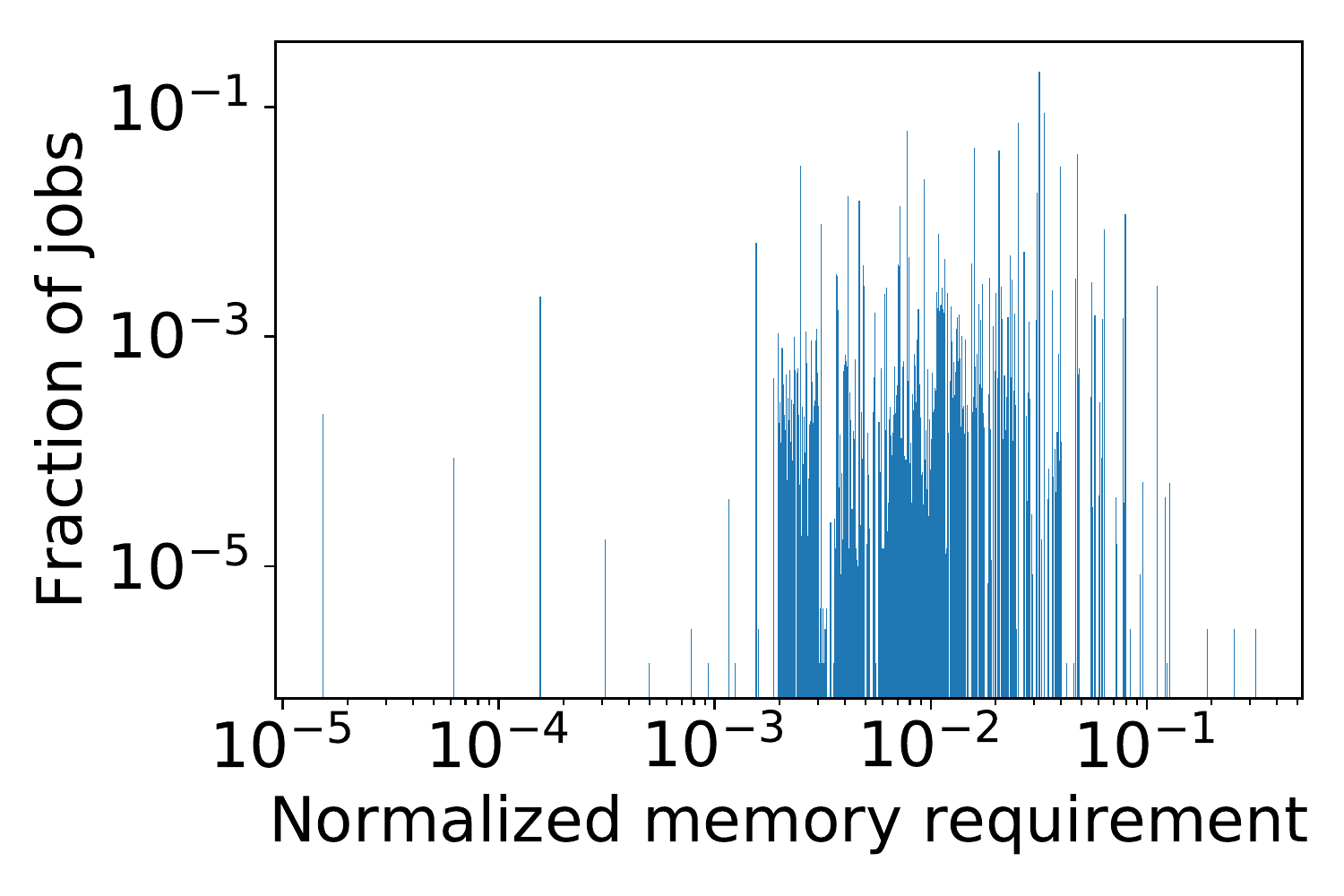}
	\label{fig:simgtracestatsmem}
    \end{subfigure}%
    \begin{subfigure}{.5\columnwidth}
  	\centering
	\includegraphics[width=\columnwidth]{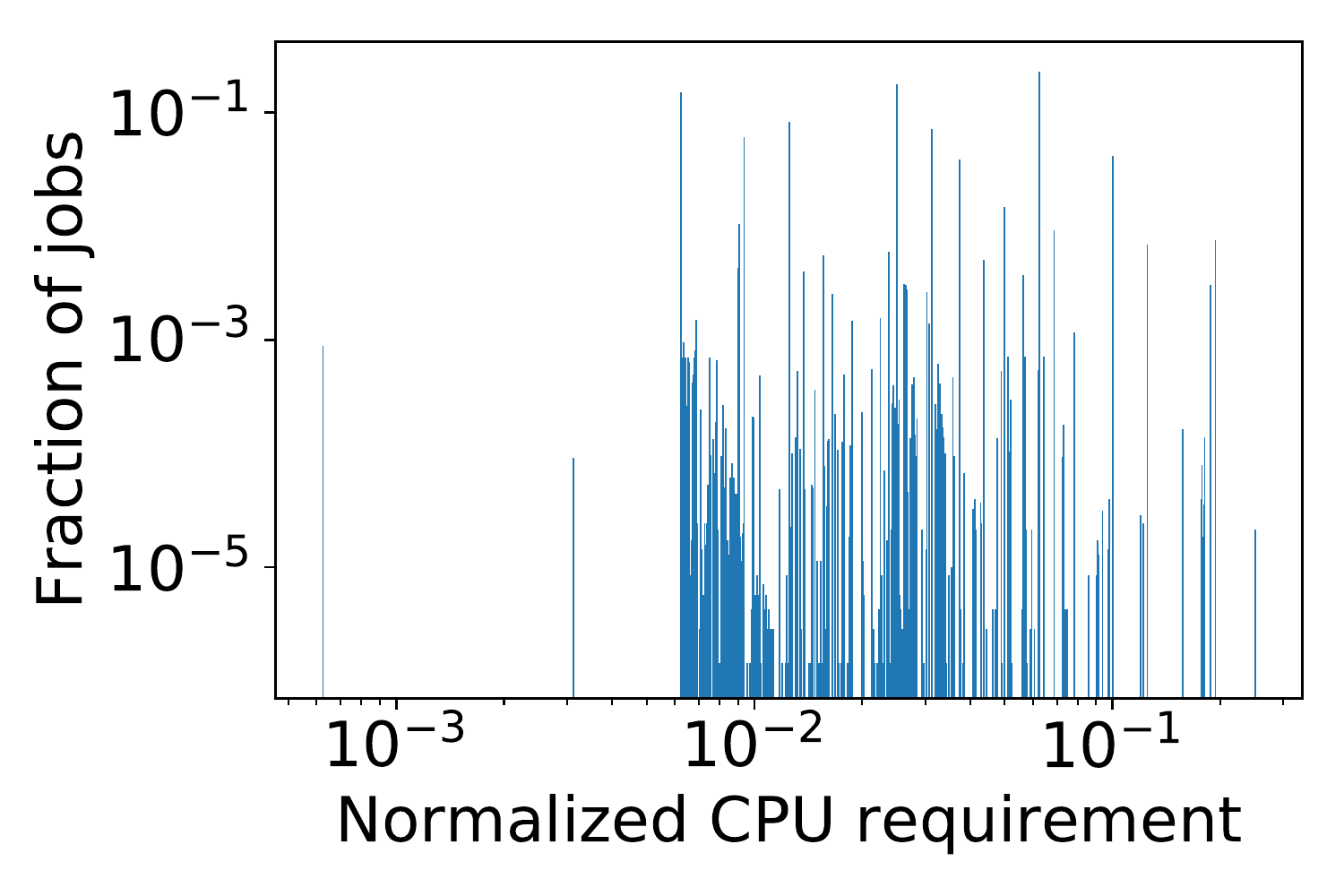}
	\label{fig:simgtracestatscpu}
    \end{subfigure}
\vspace{-0.1 in}
    \caption{There are more than 700 discrete memory 
    	requirements and 400 discrete CPU requirements in the 
        tasks submitted to a Google cluster during a day.}
    \label{fig:simgtracestats}
\end{figure} 

Despite the vast literature on scheduling algorithms, their theoretical study 
in such high-dimensional setting is very limited. The majority of the past work 
relies on a crucial assumption that there is a predefined finite set of discrete types, 
e.g.~\cite{magsriyi12, Stolyar2013, Siva13, Siva14, ghaderi2016randomized, PG17}. 
Although we can consider every possible resource profile as a type, the 
number of such types could be formidably large. The application of scheduling 
algorithms, even with polynomial complexity in the number of types, is 
discouraging in such setting. A natural solution could be to divide the 
resource requests into a smaller number of types. Such a scheduler can be 
strictly suboptimal, since, as a result of mapping to a smaller number of 
types, jobs may underutilize or overutilize the resource compared to what 
they actually require. Moreover, in the \textit{absence of any prior knowledge about 
the resource demand statistics}, it is not clear how the partitioning of the resource 
axis into a small number of types should be actually done.

Our work fulfills one of the key deficiencies of the past work in the modeling 
and analysis of scheduling algorithms for distributed server systems. 
Our model allows a very large or, in the extreme case, even \textit{infinite} number of job types, 
i.e., when the jobs' resource requirements follow a 
probability distribution over a continuous support. 
To the best of our knowledge, there is no past work on characterizing the optimal 
throughput and what can be achieved when there are no discrete job types. Our goal is 
to characterize this throughput and design algorithms that: (1) have low complexity, 
and (2) can provide provable throughput guarantees \textit{without} the knowledge 
of the traffic or the resource requirement statistics.

\subsection{Related Work}
Existing algorithms for scheduling jobs in distributed computing platforms 
can be organized in two categories.

In the first category, we have algorithms that do not provide any throughput guarantees, 
but perform well empirically or focus on other performance metrics such as fairness 
and makespan. These algorithms include slot-based schedulers that divide 
servers into a predefined number of slots for placing tasks~\cite{isard2009quincy,tang2013dynamic},
resource packing approaches such as~\cite{grandl2015multi,Verma2015}, fair resource sharing 
approaches such as~\cite{Ghodsi2011,chowdhury2016hug}, and Hadoop's default schedulers such 
as FIFO~\cite{usama2017job}, Fair scheduler~\cite{Fairscheduler}, and Capacity 
scheduler~\cite{Capscheduler}.

In the second category, we have schedulers with throughput guarantees, 
e.g.,~\cite{magsriyi12, Siva13, Siva14, ghaderi2016randomized, PG17}. 
They work under the assumption that there is a finite number of discrete job types. 
This assumption naturally lends itself to \textit{MaxWeight} algorithms~\cite{tassiulas1992stability}, where each server schedules jobs according to a maximum weight configuration chosen from a finite set of configurations. 
The number of configurations however grows exponentially 
large with the number of types, making the application of these algorithms discouraging 
in practice. Further, their technique \textit{cannot} be applied to our setting 
which can include an infinite number of job types.

There is also literature on classical bin packing problem~\cite{garey2}, where given a list of objects of various sizes, and an infinite number of unit-capacity bins, the goal is to use the minimum number of bins to pack the objects. Many algorithms have been proposed for this problem with approximation ratios for the optimal number of bins or waste, e.g.~\cite{john,kenyon2002linear,coffman1996approximation}. 
There is also work in a setting of bin packing with 
queues, e.g.~\cite{Shah2008,coffman2001bandwidth,gamarnik2004stochastic}, under the model that an empty bin arrives at each time, then some jobs from the queue are 
packed in the bin at that time, and the bin cannot be reused in future.  
Our model is \textit{fundamentally} different from these lines of work, as the number of servers (bins) in our setting is fixed and we need to reuse the servers to schedule further jobs from the queue, when jobs depart from servers.  

\subsection{Main Contributions}
Our main contributions can be summarized as follows:  
\begin{itemize}[leftmargin=4mm]
 	\item [\bf 1.] \textbf{Characterization of Maximum Achievable Throughput.}
		We characterize the maximum throughput (\textit{maximum supportable workload}) that can be 
		theoretically achieved by any scheduling algorithm in the setting that the 
		jobs' resource requirements follow a general probability distribution $F_R$ over possibly 
		infinitely many job types. The construction of optimal schedulers to approach this 
		maximum throughput relies on a careful partition of jobs into sufficiently large number of 
		types, using the complete knowledge of the resource probability distribution $F_R$.
 	\item [\bf 2.] \textbf{Oblivious Scheduling Algorithms.}
		We introduce scheduling algorithms based on ``\textit{Best-Fit}'' packing and ``\textit{universal partitioning}'' of resource requirements into types, \textit{without} the knowledge of the resource probability distribution $F_R$. The algorithms have low complexity and can provably achieve at least $1/2$ and 
		$2/3$ of the maximum throughput, respectively. Further, we show that $2/3$ is tight 
		in the sense that no oblivious scheduling algorithm, that maps the resource requirements 
		into a finite number of types, can achieve better than $2/3$ of the maximum throughput 
		for all general resource distributions $F_R$.   
	\item[\bf 3.] \textbf{Empirical Evaluation.}
		We evaluate the throughput and queueing delay performance of all algorithms empirically
    	using both synthetic and real traffic traces.
\end{itemize}

\section{System Model and Definitions}
\subsubsection*{Cluster Model}
We consider a collection of $L$ servers denoted by the set $\mathcal{L}$.
For simplicity, we consider a single resource (e.g. memory) and assume that the servers have the same resource capacity. 
While job resource requirements are in general multi-dimensional (e.g. CPU, memory), it has been
observed that memory is typically the bottleneck resource~\cite{Capscheduler,nitu2018working}.
Without loss of generality, we assume that each server's capacity is normalized to one. 
\subsubsection*{Job Model} Jobs arrive over time, and the $j$-th job, $j=1,2,\cdots$, requires an amount $R_j$ of the (normalized) resource for the duration of its service. 
The resource requirements $R_1, R_2,\cdots$ are i.i.d. random variables with a 
\textit{general} cdf (cumulative distribution function) $F_R (\cdot): (0,1] \to [0,1]$, with average $\bar{R} = \mathds{E}(R)$.
Note that each job should be served by one server and its resource requirement \textit{cannot be 
fragmented} among multiple servers. 
In the rest of the paper, we use the terms job size and job resource requirement interchangeably.  

\subsubsection*{Queueing Model}
We assume time is divided into time slots $t=0,1,\cdots$. 
At the beginning of each time slot $t$, a set $\calA(t)$ of jobs arrive to the system. 
We use $A(t)$ to denote the cardinality of $\calA(t)$. 
The process $A(t)$, $t=0,1,\cdots$, is assumed to be i.i.d. with a finite mean 
$\mathds{E}[A(t)]=\lambda$ and a finite second moment. 

There is a queue $\calQ(t)$ that contains the jobs that have arrived up to time slot $t$ 
and have not been served by any servers yet. At each time slot, the scheduler 
can select a set of jobs $\calD(t)$ from $\calQ(t)$ and place each job in a server that has 
enough available resource to accommodate it.
Specifically, define $\calH(t)= (\calH_\ell(t),\ \ell \in \calL)$, where 
$\calH_\ell(t)$ is the set of existing jobs in server $\ell$ at time $t$. At any time, the total size of the jobs packed in server $\ell$ \textit{cannot} exceed its capacity, 
i.e., 
\be \label{eq:packing}
\sum_{j \in \calH_\ell(t)} R_j \leq 1,\ \forall \ell \in \calL, \ t=0,1,\cdots.
\ee 
Note that jobs may be scheduled out of the order that they arrived, depending on the resource availability of servers. Let $D(t)$ denote the 
cardinality of $\calD(t)$ and $Q(t)$ denote the cardinality of $\calQ(t)$ 
(the number of jobs in the queue). Then the queue $\calQ(t)$ and its size $Q(t)$ evolve as
\be 
&\calQ(t+1)=\calQ(t) \cup \calA(t)-\calD(t), \label{eq:queue} \\
&Q(t+1)=Q(t)+A(t)-D(t). \label{eq:queuesize}
\ee

Once a job is placed in a server, it completes its service after a geometrically distributed amount of time with mean $1/\mu$, after which it releases its reserved resource. 
This assumption is made to simplify the analysis, and the results can be extended to more general service time distributions
(see Section~\ref{sec:dis} for a discussion).  

\subsubsection*{Stability and Maximum Supportable Workload}
The system state is given by $(\calQ(t),\calH(t))$ which evolves 
as a Markov process over an \textit{uncountably infinite state space}
\footnote{The state space can be equivalently represented in a complete separable metric space, as we show in Section~\ref{pr:thm_GF1}.}
We investigate the stability of the system in terms of the average queue size, 
i.e., the system is called stable if 
$\lim \sup _{t } \mathds{E}[Q(t)] <\infty$. 
Given a job size distribution $F_R$, a workload $\rho:=\lambda/\mu$ is called supportable if there exists a scheduling policy that can stabilize the system for the job arrival rate $\lambda$ and the mean service duration $1/\mu$. 

\textit{Maximum supportable workload} is a workload $\rho^\star$ such that any 
$\rho <\rho^\star$ can be stabilized by some scheduling policy, which possibly uses the knowledge of the job size distribution $F_R$,  
but no $\rho >\rho^\star$ can be stabilized by any scheduling policy.  

\section{Characterization of Maximum Supportable Workload}\label{sec:stab}
In this section, we provide a framework to characterize the maximum supportable workload $\rho^\star$ given a job resource distribution $F_R$. We start with an overview of the results for a system with a finite set of discrete job types.
\subsection{Finite-type System}
It is easy to characterize the maximum supportable workload when jobs belong to a finite set of discrete types.
In this case, it is well known that the supportable workload region is the sum of  convex hull of \textit{feasible configurations} of servers, e.g.~\cite{magsriyi12, Siva13, Siva14, ghaderi2016randomized, PG17},
which are defined as follows.
\begin{definition}[Feasible configuration]\label{def:conf}
Suppose there is a finite set of $J$ job types, with job sizes $r_1, \cdots, r_J$. An integer-valued vector $\mathbf{k} = \left(k_1, \cdots, k_{J}\right)$ is a feasible configuration for a server if it is possible to simultaneously pack $k_1$ jobs of 
of type $1$, $k_2$ jobs of type $2$, $\dots$, and $k_J$ jobs of type $J$ in the server, without exceeding its capacity. Assuming normalized server's capacity, any feasible configuration $\mathbf{k}$ must therefore satisfy $\sum_{j=1}^J k_j r_j\leq 1$, $k_j \in \mathbb{Z}_+$, $j=1,\cdots, J$. 
We use $\overline{\mathcal{K}}$ to denote the (finite) set of all feasible configurations.
\end{definition}

We define $P_j\triangleq\mathbb{P}(R=r_j)$ to be the probability that size of an arriving job is $r_j$, $\bm{P}=(P_1,\cdots,P_J)$ to be the vector of such arrival probabilities, and $\rho = \lambda/\mu$ to be the workload. We also refer to $\rho \bm{P}$ as \textit{the workload vector}. As shown in \cite{magsriyi12, Siva13, Siva14, ghaderi2016randomized}, the maximum supportable workload $\rho^\star$ is 

\be \label{eq:region}
\rho^\star= \sup \Big\{\rho \in \mathbb{R}_+ : \rho \bm{P} < \sum_{\ell \in \calL} \mathbf{x}^\ell, \mathbf{x}^\ell\in \rm Conv(\overline{\mathcal{K}}), \ell \in \calL \Big\}
\ee
where $\rm Conv(\cdot)$ is the convex hull operator, and the vector inequality is component-wise. Also $\sup$ (or $\inf$) denotes \textit{supremum} (or \textit{infimum}). Hence any $\rho < \rho^\star$ is supportable by some scheduling algorithm, while no $\rho>\rho^\star$ can be supported by any scheduling algorithm. 

The optimal or near-optimal scheduling policies then basically follow 
the well-known \textit{MaxWeight} algorithm~\cite{tassiulas1992stability}. Let $Q_j(t)$ be the number of type-$j$ jobs waiting in queue at time $t$. At any time $t$ for each server $\ell$, the algorithm maintains a feasible configuration $\mathbf{k}(t)$ that has the ``maximum weight''~\cite{Siva13, Siva14} (or a fraction of the maximum weight~\cite{PG17}), among all the feasible configurations $\overline{\mathcal{K}}$. The weight of a configuration is formally defined below. 
\begin{definition}[Weight of a configuration]\label{def:weight}
Given a queue size vector $\bm{Q}=(Q_1, \cdots, Q_J)$, the weight of a feasible configuration $\mathbf{k}= \left(k_1, \cdots, k_{J}\right)$ 
is defined as the inner product
\be
&\langle \mathbf{k}, \mathbf{Q} \rangle =
\sum_{j=1}^{J} k_j Q_j.
\ee
\end{definition}
\subsection{Infinite-type System} 
In general, the support of the job size distribution $F_R$ can span an infinite 
number of types (e.g., $F_R$ can be a continuous function over $(0,1]$). We introduce 
the notion of virtual queue which is used to characterize 
the supportable workload for any general distribution $F_R$.

\begin{definition}[Partition and Virtual Queues (\rm{VQs})]\label{def:virtual}
Define a partition ${X}$ of interval $(0,1]$ as a finite collection of disjoint subsets 
$X_j \subset (0, 1]$, $j = 1, \cdots, J$, such that $\cup_{j=1}^{J} X_j = (0, 1]$.
If the size of an arriving job belongs to $X_j$, we say it is a type-$j$ job.
For each type $j$, we consider a \textit{virtual queue} $\rm{VQ}_j$ 
which contains the type-$j$ jobs waiting in the queue for service.
\end{definition}

As in the finite-type system, given a partition $X$, we can define the probability that a type-$j$ job arrives as 
	$P^{(X)}_j \triangleq \prob{R \in X_j}$, the arrival probability vector as $\bm{P}^{(X)} = (P_1, \cdots, P_J)$, 	and the workload vector as $\rho \bm{P}^{(X)}$. However, under this definition, it is not clear what configurations are feasible, since the 
jobs in the same virtual queue can have different sizes, even though they are 
called of the same type. Hence we make the following definition.

\begin{definition}[Rounded \rm{VQs}]\label{def:rounded}
We call $\rm{VQs}$ ``\textit{upper-rounded $\rm{VQs}$}'', if the sizes of type-$j$ jobs are assumed to be $r_j = \sup{X_j}$, $j=1,\cdots, J$.  
Similarly, we call them ``\textit{lower-rounded $\rm{VQs}$}'', if the sizes of type-$j$ jobs are assumed to be $r_j = \inf{X_j}$, $j=1,\cdots, J$.  
\end{definition}

Given a partition ${X}$, let $\overline{\rho}^\star({X})$ and 
$\underline{\rho}^\star({X})$ be respectively the maximum workload $\lambda/\mu$ 
under which the system with upper-rounded virtual queues and the system 
with the lower-rounded virtual queues can be stabilized. Since these systems have finite types, these quantities can be described by  \dref{eq:region} applied to the corresponding finite-type system with workload vector $\rho \bm{P}^{(X)}$.  

Let also $\overline \rho^\star = \sup_{{X}} \overline \rho^\star({X})$ 
and $\underline \rho^\star = \inf_{{X}} \underline \rho^\star({X})$ 
where the supremum and infimum are over all possible partitions of interval 
$(0,1]$. Next theorem states the result of existence of maximum supportable
workload.

\begin{theorem}\label{thm:OPT}
	Consider any general (continuous or discontinuous) probability 
	distribution of job sizes with cdf $F_R(\cdot)$.
	Then there exists a unique $\rho^\star$ such that 
	$\overline \rho^\star = \underline \rho^\star = \rho^\star$. Further, given any 
	$\rho < \rho^\star$, there is a partition ${X}$ such that the associated 
	upper-rounded virtual queueing system (and hence the original system) can be 
	stabilized. 
\end{theorem}

\begin{proof}  
The proof of Theorem~\ref{thm:OPT} has two steps. First, we show that  
$\overline{\rho}^\star({X}) \leq \rho^\star \leq 
\underline{\rho}^\star({X})$ for any partition $X$.
Second, we construct a sequence of partitions, that depend on the job size distribution $F_R$, and become increasingly finer, such that the difference between the two bounds vanishes in the limit. 

Full proof can be found in Appendix~\ref{pr:thm_OPT}.
\end{proof}

Theorem~\ref{thm:OPT} implies that there is a way of mapping the job sizes to a finite number of types using partitions, such that by using finite-type scheduling algorithms, 
the achievable workload approaches the optimal workload as partitions become finer.
However, the construction of the 
partition crucially \textit{relies on the knowledge of the job size
distribution} $F_R$, which may not be readily available in practice.
Further, the number of feasible configurations \textit{grows exponentially} large as the number of 
subsets in the partition increases, which prevents efficient implementation 
of discrete type scheduling policies (e.g. MaxWeight) in practice. 

Next, we focus on low-complexity scheduling algorithms 
that \textit{do not} assume the knowledge of $F_R$ a priori, and can provide 
a fraction of the maximum supportable workload $\rho^\star$.

\section{Best-Fit Based Scheduling}
The \textit{Best-Fit} algorithm was first introduced as a heuristic for \textit{Bin Packing} problem~\cite{garey2}: given a list of objects of various sizes, we are asked to pack them into bins of unit capacity so as to minimize the number of bins used. Under Best-Fit, the objects are processed one by one and each object is placed in the ``tightest'' bin (with the least residual capacity) that can accommodate
the object, otherwise a new bin is used. Theoretical guarantees of Best-Fit in terms of approximation ratio have been extensively studied under discrete and continuous object size distributions~\cite{john,kenyon2002linear,coffman1996approximation}.  

There are several \textit{fundamental} differences between the 
classical bin packing problem and our problem. 
In the bin packing problem, there is an infinite number 
of bins available and once an object is placed in a bin, it remains in the bin forever, 
while in our setting, the number of bins (the equivalent of servers) is fixed, and bins 
have to be reused to serve new objects from the queues as objects depart from the bins, 
and new objects arrive to the queue.
Next, we describe how Best-Fit (\textit{BF}) can be adapted for job scheduling in our setting.

\subsection{\textrm{BF-J/S} Scheduling Algorithm}

Consider the following two adaptations of Best-Fit (BF) for job scheduling:

\begin{itemize}[leftmargin=4mm]
	\item \textbf{BF-J} (\textit{Best-Fit from Job's perspective}): 
	 
	List the jobs in the queue in an arbitrary order (e.g. according to their arrival times). Starting from the first job, each job is placed in the server with the ``least residual capacity'' among the servers that can accommodate it, if possible, otherwise the job remains in the queue. 
	
	\item \textbf{BF-S} (\textit{Best-Fit from Server's perspective}): 
	
	List servers in an arbitrary order (e.g. according to their index). Starting from the first server, each server is filled iteratively by choosing the ``largest-size job'' in the queue that can fit in the server, until no more jobs can fit.
\end{itemize}

\textrm{BF-J} and \textrm{BF-S} need to be performed in every time slot.  	
Under both algorithms, observe that no further job from the queue can be added in any of the servers.
However, these algorithms are not computationally efficient as they both make many redundant searches over the jobs in the queue or over the servers, when there are no new job arrivals to the queue or there are no job departures from some servers. 
Combining both adaptations, we describe the algorithm below which is computationally more efficient. 
\begin{itemize}[leftmargin=4mm]
\item \textbf{BF-\textrm{J/S}} (\textit{Best-Fit from Job's and Server's perspectives}):

It consists of two steps:		
\begin{itemize}[leftmargin=2mm]
	\item [1)] 
	Perform \textrm{BF-S} only over the list of servers that had job 
	departures during the previous time slot. Hence, some jobs that have not been scheduled in the previous time slot or some of
	newly arrived jobs are scheduled in servers.	
	\item [2)] 
	Perform \textrm{BF-J} only over the list of newly arrived jobs that have not been scheduled in the first step. 
\end{itemize}
\end{itemize}

\subsection{Throughput Guarantee}
The following theorem characterizes the maximum supportable workload under \algA{}-\textrm{J/S}.

\begin{theorem}\label{thm:GF1}
	Suppose any job has a minimum size $u$. Algorithm \algA{}-\textrm{J/S} can achieve at least $\frac{1}{2}$ of the maximum supportable workload $\rho^\star$, for any $u > 0$. 
\end{theorem}
\begin{proof}
	We present a sketch of the proof here and provide the full proof in Appendix~\ref{pr:thm_GF1}.
	The proof uses Lyapunov analysis for Markov chain  $(\calQ(t),\calH(t))$ whose state includes the jobs in queues and servers and their sizes. 
	The Markov chain can be equivalently represented in a Polish space and we  prove its positive recurrence using a multi-step Lyapunov technique~\cite{Foss2004}
	and properties of \algA{}-\textrm{J/S}. We use a Lyapunov function which is the sum of sizes of all jobs in the system at time $t$. Given that jobs have a minimum size, keeping the total size bounded implies
	the number of jobs is also bounded.
	
	The key argument in the proof is that by using \algA{}-\textrm{J/S} as described, 
	all servers operate in more than ``half full'', most of the time, 
	when the total size of jobs in the queue becomes large. 
	To prove this, we consider two possible cases:
	\begin{itemize}[leftmargin=4mm]
		\item \textit{The total size of jobs in queue with size $\le \frac{1}{2}$ is large:}\\
			In this case, these jobs will be scheduled greedily whenever the server is more than half empty. Hence, the server will always become 
			more than half full until there are no such jobs in the queue.
		\item \textit{The total size of jobs in queue with size $>\frac{1}{2}$ is large:}\\ If at time slot $t$, a job in server is not completed, it will complete its service 
			within the next time slot with probability $\mu$, independently of the other 
			jobs in the server. Given the minimum job size, the number of jobs in a 
			server is bounded so it will certainly empty in 
			a finite time. Once this happens, jobs will be scheduled starting 
			from the largest-size one, and the server will remain more than half full, as long as
			there is a job of size more than $1/2$ to replace it. This step is true because 
			of the way Best-Fit works and \textit{does not} 
			hold for other bin packing algorithms like First-Fit. 
			\end{itemize}
See the full proof in Appendix~\ref{pr:thm_GF1}.
\end{proof}

\section{Partition Based Scheduling}

\algA{}-\textrm{J/S} demonstrated an algorithm that can achieve at least half of the maximum workload $\rho^\star$, without relying on any partitioning of jobs into types. In this section, we propose partition based scheduling algorithms that can provably
achieve a larger fraction of the maximum workload $\rho^\star$, using a \textit{universal partitioning} into a small number of types, without the knowledge of job size distribution $F_R$. 

\subsection{Universal Partition and Associated Virtual Queues}
Consider a partition of the interval $(1/{2^J},1]$ into the following $2J$ subintervals: 
\begin{equation}\label{eq:VQAintervals}
\begin{aligned}
&I_{2m} = \Big(\frac{2}{3} \frac{1}{2^m}, \frac{1}{2^m}\Big],\ m=0,\cdots, J-1 \\
&I_{2m+1} = \Big(\frac{1}{2} \frac{1}{2^m}, \frac{2}{3} \frac{1}{2^m}\Big] ,\ m=0,\cdots, J-1.
\end{aligned}
\end{equation}
We refer to this partition as partition $I$, where $J>1$ is a fixed parameter to be determined shortly. The odd and even subintervals in $I$ are geometrically shrinking. Figure~\ref{fig:intrvl} gives a visualization of this partition.
\begin{figure}
	\centering
	\includegraphics[width=0.95\columnwidth]{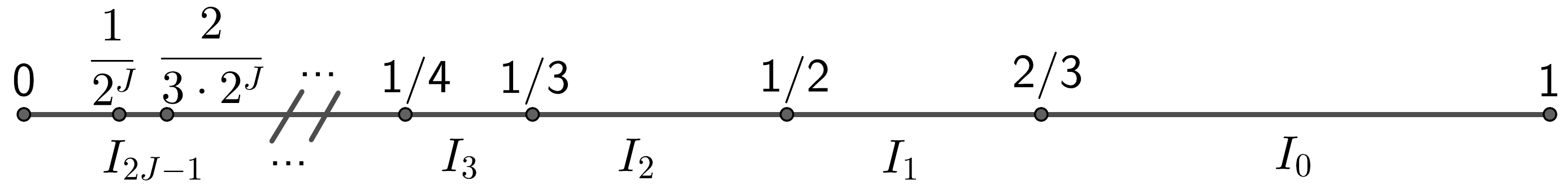}
	\caption{Partition $I$ of interval $(1/{2^J}, 1]$ based on \dref{eq:VQAintervals}.}
	\label{fig:intrvl}
\end{figure}

Jobs in queue are divided among virtual queues (Definition~\ref{def:virtual}) 
according to partition $I$. Specifically,
when the size of a job falls in the subinterval 
$I_j$, $j=0,\cdots,2J-1$, we say this job is of type $j$ and it is placed 
in a virtual queue $\mathrm{VQ}_j$, \textit{without rounding its size}. Moreover, 
jobs whose sizes fall in  $(0, 1/2^J]$ are placed in the last virtual 
queue $\mathrm{VQ}_{2J-1}$, and their sizes are rounded up to $1/2^J$.

We use $Q_{j}(t)$ to denote the size (cardinality) of $\mathrm{VQ}_j$ 
at time $t$ and use $\mathbf{Q}(t)$ to denote the vector of all VQ sizes.

\subsection{\textrm{VQS} (\textit{Virtual Queue Scheduling}) Algorithm}
To describe the \textrm{VQS} algorithm, we define the following reduced set of configurations 
which are feasible for the system of upper-rounded $\mVQ$s (Definition~\ref{def:rounded}) 
\begin{definition}[Reduced feasible configuration set] \label{def:reduced}
	The \textit{reduced feasible configuration set}, denoted by ${\mathcal{K}}^{(J)}_{RED}$, consists of the following $4J-4$ configurations: 
	\begin{equation}\label{eq:configs}
	\begin{aligned}
	2^m \mathbf{e}_{2m}, \quad &m = 0, \cdots, J-1 \\
	3 \cdot 2^{m-1} \mathbf{e}_{2m+1}, \quad &m = 1, \cdots, J-1 \\
	\mathbf{e}_{1} + \lfloor 2^m/3 \rfloor \mathbf{e}_{2m}, 
	\quad &m = 2, \cdots, J-1 \\
	\mathbf{e}_{1} + 2^{m-1} \mathbf{e}_{2m+1}, 
	\quad &m = 1, \cdots, J-1 \\
	\end{aligned}
	\end{equation}
	where $\mathbf{e}_{j} \in \mathbb{Z}^{2J}$ denotes the basis vector 
	with a single job of type $j$, $j=0,\cdots,2J-1$, and zero jobs of any other types.
\end{definition} 
Note that each  configuration $\mathbf{k}=(k_0,\cdots,k_{2J-1}) \in {\mathcal{K}}^{(J)}_{RED}$ either contains jobs from only one $\mVQ_j$, $j=0,\cdots, 2J-1$, or contains   
jobs from $\mVQ_1$ and one other $\mVQ_j$.

The ``\textit{VQS algorithm}'' consists of two steps: (1) setting active configuration, and (2) job scheduling using the active configuration:
\begin{itemize}[leftmargin=5mm]
	\item [1.]\textit{Setting active configuration}: 
	
	Under $\mathrm{VQS}$, every server $\ell \in \calL$ has an \textit{active 
		configuration} $\mathbf{k}^\ell (t) \in {\mathcal{K}}^{(J)}_{RED}$ which is renewed only when the server becomes empty.
	Suppose time slot $\tau^{\ell}_i$ is the $i$-th time that server $\ell$ is empty (i.e., it has been empty or all its jobs depart during this time slot). At this time, the configuration of server $\ell$ is set to the max weight configuration among the configurations of $\mathcal{K}^{(J)}_{RED}$ (Definitions~\ref{def:weight} and \ref{def:reduced}), i.e., 
	\be
	\mathbf{k}^\star (\tau^{\ell}_i)= 
	\argmax_{\mathbf{k} \in \mathcal{K}_{RED}^{(J)}}
	{\langle \mathbf{k}, \mathbf{Q}(\tau^{\ell}_i) \rangle} =\argmax_{\mathbf{k} \in \mathcal{K}_{RED}^{(J)}}
		\sum_{j=0}^{2J-1} k_{j} Q_{j}.
	\ee   
The active configuration remains fixed until the next time $\tau^{\ell}_{i+1}$ that the server becomes empty gain, i.e., 
\be
	\mathbf{k}^{\ell}(t)=\mathbf{k}^\star (\tau^{\ell}_i),\ \tau^{\ell}_i \leq t < \tau^{\ell}_{i+1}.
\ee  
	\item [2.] \textit{Job scheduling:}
	
	Suppose the active configuration of server $\ell$ at time $t$ is $\mathbf{k}\in {\mathcal{K}}^{(J)}_{RED}$. Then the server  schedules jobs as follows:
	\begin{enumerate}[leftmargin=3mm]
		\item [(i)] 
		If $k_{1}=1$, the server reserves $2/3$ of its capacity for 
		serving jobs from $\mVQ_1$, so it can serve at most one job
		of type $1$ at any time. 
		If there is no such job in the server already, it schedules one from $\mVQ_1$.
		\item [(ii)] Any configuration $\mathbf{k} \in {\mathcal{K}}^{(J)}_{RED}$ 
		has at most one $k_{j} > 0$ other than $k_{1}$. 
		The server will 
		schedule jobs from the corresponding $\mVQ_j$, starting from the 
		head-of-the-line job in $\mVQ_j$, until no more jobs can fit in the server. 
		The actual number of jobs scheduled from $\mVQ_j$ in the server could 
		be more than $k_{j}$ depending on their actual sizes.
	\end{enumerate}   
\end{itemize}  
\begin{remark}
The reason for choosing times $\tau_i^{\ell}$ to renew the configuration of server $\ell$ is to \textit{avoid possible preemption} of existing jobs in server (similar to~\cite{magsriyi12,Siva14}). Also note that active configurations in ${\mathcal{K}}^{(J)}_{RED}$ are based on upper-rounded $\mVQ$s. Since jobs \textit{are not} actually rounded in $\mVQ$s, the algorithm can schedule more jobs than what specified in the configuration.
\end{remark}

\subsection{Throughput Guarantee}
The {\algB} algorithm can provide a stronger throughput guarantee than \textrm{BF-J/S}. 
A key step to establish the throughput guarantee is related
to the property of configurations in the set $\mathcal{K}^{(J)}_{RED}$, which is stated below.
\begin{proposition}\label{prop:2/3} 
	Consider any partition $X$ which is a refinement of partition $I$, 
	i.e., any subset of $X$ is contained in an interval $I_j$ in \dref{eq:VQAintervals}. 
	Given any set of jobs with sizes in $(1/2^J,1]$ in the queue, let $\bm{Q}$ and $\mathbf{Q}^{(X)}$ be the corresponding vector of $\mVQ$ sizes under partition $I$ and partition $X$. Then there is a configuration $\mathbf{k} \in \mathcal{K}_{RED}^{(J)}$ such that
    \begin{equation}\label{eq:23_opt}
	\langle {\mathbf{k}}, \mathbf{Q}\rangle\ge \frac{2}{3}
	\langle {\mathbf{k}}^{(X)}, \mathbf{Q}^{(X)} \rangle,\ \forall {\mathbf{k}}^{(X)} \in {\mathcal{K}}^{(X)},
	\end{equation}
where ${\mathcal{K}}^{(X)}$ is the set of ``all'' feasible configurations based on  \textit{upper-rounded} $\mVQ$s for partition $X$. 
\end{proposition}

\begin{proof}
	
	For simplicity of description, consider $X$ to be a partition of 
	$(1/2^J, 1]$ into $N$ subintervals $(\xi_{i-1}, \xi_{i}]$, $i=1,\cdots, N$.
	The proof arguments are applicable to any other types of subsets of
	$(1/2^J, 1]$ as long as each subset is contained 
	in an interval $I_j$ in \dref{eq:VQAintervals}.

	Given the proposition's assumption, we can define sets $Z_j$,
	$j = 0,\cdots, 2J-1$, such that 
	$i \in Z_{j}$ iff $\xi_{i} \in I_j$.
	Any job in $\mVQ^{(X)}_i$, $i \in Z_{j}$, under partition $X$, belongs to $\mVQ_j$ under partition $I$, therefore 
	\begin{equation}\label{eq:Z}
	\sum_{i \in Z_{j}} Q_i^{(X)} = Q_{j}.
	\end{equation}
	
	Let $\langle {\mathbf{k}}^{(X)}, \mathbf{Q}^{(X)} \rangle = U$. 
	Note that in any feasible configuration 
	${\mathbf{k}}^{(X)} \in {\mathcal{K}}^{(X)}$, $\sum_{i \in Z_{1}} {k}^{(X)}_{i}$ can be $0$ or $1$. 
	To show \dref{eq:23_opt}, we consider these two cases separately:
	
	\noindent \textbf{Case 1.} $\sum_{i \in Z_{1}} {k}^{(X)}_{i} = 0$:
	
	\noindent We claim at least one of the following inequalities is true
	\begin{equation}\label{eq:Q_U}
	\begin{aligned}
	Q_{2m} \ge 2U/3 \times 1/2^{m}, \quad m=0,\cdots, J-1 \\
	Q_{2m+1} \ge U/2 \times 1/2^{m}, \quad m=1,\cdots, J-1
	\end{aligned}
	\end{equation}
	
	If the claim is not true, we reach a contradiction because
	\ben
	\begin{aligned}
		&U = 
		\sum_{m=0}^{J-1} \sum_{i_0 \in Z_{2m}} {k}_{i_0}^{(X)} Q_{i_0}^{(X)} +
		\sum_{m=1}^{J-1} \sum_{i_1 \in Z_{2m+1}} {k}_{i_1}^{(X)} Q_{i_1}^{(X)} \stackrel{(a)}{<} \\
		&\Big( \sum_{m=0}^{J-1} \sum_{i_0 \in Z_{2m}} {k}_{i_0}^{(X)} \frac{2}{3}\frac{1}{2^m} +
		\sum_{m=1}^{J-1} \sum_{i_1 \in Z_{2m+1}} {k}_{i_1}^{(X)} \frac{1}{2} \frac{1}{2^{m}} 
		\Big) U \stackrel{(b)}{<}\\
		&\Big(\sum_{m=0}^{J-1} \sum_{i_0 \in Z_{2m}} {k}_{i_0}^{(X)} \xi_{i_0} + 
		\sum_{m=1}^{J-1} \sum_{i_1 \in Z_{2m+1}} {k}_{i_1}^{(X)} \xi_{i_1} 	\Big) U	\stackrel{(c)}{\le} 1 \times U \label{eq:contradiction1},
	\end{aligned}
	\een
	where $(a)$ is due to the assumption that none of inequalities in \dref{eq:Q_U} hold and using the fact that $Q_{i}^{(X)} \le Q_{j}$ if $i \in Z_{j}$, $(b)$ is due to
	the fact $\xi_{i} > \inf{I_j}$ if $i \in Z_j$, 
	and $(c)$ is due to the server's capacity constraint for feasible configuration ${\mathbf{k}}^{(X)}$.
	
	Hence, one of the inequalities in \dref{eq:Q_U} must be true. 
	If $Q_{2m} \ge 2U/3 \times 1/2^{m}$ for some $m =0, \cdots, J-1$,
	then~(\ref{eq:23_opt}) is true for configuration
	${\mathbf{k}} = 2^m \mathbf{e}_{2m}$,
	while if $Q_{2m+1} \ge U/2 \times 1/2^{m}$ for some $m = 1, \cdots, J-1$,
	then (\ref{eq:23_opt}) is true for configuration
	${\mathbf{k}} = 3 \cdot 2^{m-1} \mathbf{e}_{2m+1}$.

	\noindent \textbf{Case 2.} $\sum_{i \in Z_{1}} {k}^{(X)}_{i} = 1$:\\
	In this case $\sum_{i \in Z_{0}} {k}^{(X)}_i = 0$.
	We further distinguish three cases for $Q_{1}$ compared to $U$: 
	$Q_{1} \ge \frac{2U}{3}$, $\frac{2U}{3} > Q_{1} \ge \frac{U}{2}$, and $\frac{U}{2} > Q_{1}$.
	In the second case, we further consider two subcases depending on 
	$\sum_{i \in Z_{2}} k_{i}^{(X)}$ being $0$ or $1$.
	Here we present the analysis of the case $\frac{2U}{3} > Q_{1} \ge \frac{U}{2}$,  
	$\sum_{i \in Z_{2}} k_{i}^{(X)}=0$. 
	The rest of the cases are either trivial or follow a similar argument
	and can be found in Appendix~\ref{pr:prop_2/3}.
	
	Let $U^\prime := U - Q_{1}$, then one of the following inequalities has to be true
	\begin{equation}\label{eq:Q_Up}
	\begin{aligned}
	Q_{2m} \ge U^\prime/(3 \cdot 2^{m-2}),\ m = 2, \cdots J-1\\
	Q_{2m+1} \ge U^\prime/(3 \cdot 2^{m-1}),\ m = 1, \cdots J-1,\\
	\end{aligned}
	\end{equation}
	otherwise, we reach a contradiction, similar to Case 1, i.e.,
	\ben
	\label{eq:contradiction221}
	\begin{aligned}
		&U^\prime =
		\sum_{m=2}^{J-1} \sum_{i_0 \in Z_{2m}} {k}_{i_0}^{(X)} Q_{i_0}^{(X)} +
		\sum_{m=1}^{J-1} \sum_{i_1 \in Z_{2m+1}} {k}_{i_1}^{(X)} Q_{i_1}^{(X)} \stackrel{(a)}{<}\\
		& 2 U^\prime \Big( \sum_{m=2}^{J-1} \sum_{i_0 \in Z_{2m}}{k}_{i_0}^{(X)} \frac{2}{3} \frac{1}{2^m} +
		\sum_{m=1}^{J-1}\sum_{i_1 \in Z_{2m+1}} {k}_{i_1}^{(X)} \frac{1}{3} \frac{1}{2^m} \Big) \stackrel{(b)}{<} U^\prime
	\end{aligned}
	\een
	where $(a)$ is due to the assumption that none of inequalities in 
	\dref{eq:Q_Up} hold, and $(b)$ is due to the constraint that the jobs
	in the configuration $\mathbf{k}^{(X)}$, other than the job types in $Z_1$, 
	should fit in a space of at most $1/2$ (the rest is occupied by a job of size at least $1/2$). 
	It is then easy to verify that if $Q_{2m} \ge U^\prime/(3 \cdot 2^{m-2})$
	for some $m \in [2, \cdots J-1]$ then inequality \dref{eq:23_opt} is true
	for configuration $\mathbf{e}_{1} + \lfloor 2^m/3 \rfloor \mathbf{e}_{2m}$ as
	
	\begin{equation}
	\begin{aligned}
	\langle \mathbf{k}, \mathbf{Q} \rangle &=
	Q_{1} + \lfloor 2^m/3 \rfloor Q_{2m} \ge
	Q_{1} + 2^{m-2} Q_{2m} \\
	&\ge Q_{1} + U^\prime/3 \ge 2Q_{1}/3 + U/3 \ge 2U/3
	\end{aligned}
	\end{equation}
	
	Similarly if $Q_{2m+1} \ge U^\prime/(3 \cdot 2^{m-1})$
	for some $m \in [1, \cdots J-1]$ then inequality \dref{eq:23_opt} is true
	for configuration $\mathbf{e}_{1} + 2^{m-1} \mathbf{e}_{2m+1}$ as
	
	\begin{equation}
	\begin{aligned}
	\langle \mathbf{k}, \mathbf{Q} \rangle &= 
	Q_{1} + 2^{m-1} Q_{2m+1} \\
	& \ge Q_{1} + U^\prime/3 \ge 2Q_{1}/3 + U/3 \ge 2U/3
	\end{aligned}
	\end{equation}
\end{proof}

The following theorem states the result regarding throughput of \algB{}.
\begin{theorem}\label{thm:VQS1}
	\algB{} achieves at least $\frac{2}{3}$ of the optimal workload $\rho^\star$, 
    if arriving jobs have a minimum size of at least $1/2^J$.
\end{theorem}
\begin{proof}
	The proof uses Proposition~\ref{prop:2/3} and multi-step Lyapunov technique (Theorem 1 of~\cite{Foss2004}), 
	The proof can be found in Appendix~\ref{pr:thm_VQS1}.
\end{proof}

Hence, given a minimum job's resource requirement $u>0$, $J$ has to be chosen larger than $\log_2 (1/u)$ in the VQS algorithm. Theorem~\ref{thm:VQS1} is not trivial as it implies that 
by scheduling under the configurations in ${\mathcal{K}}^{(J)}_{RED}$ (\ref{eq:configs}), 
on average 
at most $1/3$ of each server's capacity will be underutilized because of 
capacity fragmentations, \textit{irrespective} of the job size distribution $F_R$. Moreover, using ${\mathcal{K}}^{(J)}_{RED}$ reduces the search space from $O(\mathrm{Exp}(J))$ configurations to only $4J-4$ configurations, while still guaranteeing $2/3$ of the optimal workload $\rho^\star$.      

A natural and less dense partition could be to only consider the cuts at points 
$1/2^j$ for $j=0, \cdots, J$. This creates a partition consisting of  $J$ 
subintervals $\widetilde{I}_j=I_{2j} \cup I_{2j+1}$. 
The convex hull of only the first $J$ configurations of $\mathcal{K}^{(J)}_{RED}$
contains all feasible configurations of this partition.
Using arguments similar to proof of Theorem~\ref{thm:VQS1}, we can show that 
this partition can only achieve $1/2$ of the optimal workload $\rho^\star$.
One might conjecture that by refining partition $I$ \dref{eq:VQAintervals} 
or using different partitions, we can achieve a fraction larger than $2/3$ of the optimal workload $\rho^\star$; 
however, if the partition is agnostic to the job size distribution $F_R$, 
refining the partition or using other partitions \textit{does not} help. 
We state the result in the following Proposition.

\begin{proposition}\label{prop:2/3opt}
	Consider any partition $X$ consisting of a finite number of disjoint sets $X_j$, $\cup_{j=1}^N X_j = (0, 1]$. 
	Any scheduling algorithm that maps the sizes of jobs in $X_j$ to $r_j= \sup {X_j}$ (i.e., schedules based on upper-rounded $\mVQ$s) 
	cannot achieve more than $2/3$ of the optimal workload $\rho^\star$ for all $F_R$. 
\end{proposition}
\begin{proof}
See Appendix~\ref{pr:2/3opt} for the proof.
\end{proof}  

Theorem~\ref{thm:VQS1} assumed that there is a minimum resource requirement of at least $1/2^J$. This assumption can be relaxed as stated in the following corollary.
\begin{corollary}\label{cor:VQS_J} 
	Consider any general distribution of job sizes $F_R$. Given any $\epsilon > 0$, choose $J$ to be the smallest integer such that $F_R(1/2^J) < \epsilon$, then the
    \algB{} algorithm achieves at least $(1-\epsilon)\frac{2}{3}$ of the optimal workload $\rho^\star$.
\end{corollary}
\begin{proof}
See Appendix~\ref{pr:VQS_J} for the proof.
\end{proof}

Since the complexity of \algB{} algorithm is linear on $J$, it is worth increasing
it if that improves maximum throughput. An implication of Corollary~\ref{cor:VQS_J} 
is that this can be done adaptively as estimate of $F_R$ becomes available.

\section{\algB{}-{\algA}: Incorporating Best-Fit in \algB}
While the {\algB} algorithm achieves in theory a larger fraction of the optimal workload 
than {\algA}-J/S, 
it is quite inflexible compared to {\algA}-J/S, as it can only schedule according to certain
job configurations and the time until configuration changes may be long, hence might cause excessive queueing delay.
We introduce a hybrid \algB{}-{\algA} algorithm that achieves the same fraction of the optimal workload as {\algB}, but in practice has the flexibility of {\algA}. The  algorithm has two steps similar to {\algB}: Setting the active configuration is exactly the same as the first step in {\algB}, but it differs in the way that jobs are scheduled in the second step. Suppose the active configuration of server $\ell$ at time $t$ is $\mathbf{k} \in {\mathcal{K}}^{(J)}_{RED}$, then:
\begin{itemize}
	\item [(i)] 
	If $k_{1}=1$, the server will try to schedule the \textit{largest-size job} 
	from $\mVQ_1$ that can fit in it. This may not be possible because of jobs 
	already in the server from previous time slots. Unlike {\algB}, when jobs from $\mVQ_1$ are scheduled, they reserve exactly the amount of resource that they require, and no amount of resource is reserved if no job from $\mVQ_1$ is scheduled.

	\item [(ii)] 
	Any configuration $\mathbf{k} \in {\mathcal{K}}^{(J)}_{RED}$ 
	has at most one $k_{j}>0$ other than $k_{1}$. 
	Server attempts to schedule jobs from the corresponding $\mVQ_j$, 
	starting from the \textit{largest-size job} that can fit in it. 
	Depending on prior jobs in server, this procedure will stop when either 
	the number of jobs from $\mVQ_j$ in the server 
	is at least $k_{j}$, or $\mVQ_j$ becomes empty, or no more jobs 
	from $\mVQ_j$ can fit in the server. 
	\item [(iii)] Server uses \textrm{BF-S} to possibly schedule more jobs in its remaining capacity from the remaining jobs in the queue.		
\end{itemize}

The performance guarantee of \algB{}-{\algA} is the same as that of
\algB{}, as stated by the following theorem. 
\begin{theorem}\label{thm:VQS2}
	If jobs have a minimum size of at least $1/2^J$,
	\algB{}-{\algA} achieves at least $\frac{2}{3} \rho^\star$. Further, for a general job-size distribution $F_R$, if $J$ is chosen such that $F_R(1/{2^J}) < \epsilon$, then \algB{}-{\algA} achieves at least
	$(1-\epsilon)\frac{2}{3} \rho^\star$.
\end{theorem}
\begin{proof}
	The proof is similar to that of Theorem~\ref{thm:VQS1}. However, the difference 
	is that the configuration of a server (jobs residing in a server) is not 
	predictable, unless it empties, at which point we can ensure that it will schedule
	at least the jobs in the max weight configuration assigned to it, for
	a number of time slots proportional to the total queue length.
	The fact that the scheduling starts from the largest job in a virtual queue
	is important for this assertion, similarly to the importance of Best Fit in 
	the proof of Theorem~\ref{thm:GF1}.
	
	In case $J$ is chosen such that $F_R(1/2^J) < \epsilon$, the arguments
	in Corollary~\ref{cor:VQS_J} are applicable here as well.
	
	The full proof is provided in Appendix~\ref{pr:thm_VQS2}.
\end{proof}
\section{Evaluation Results}\label{sec:simulations}
\subsection{Synthetic Simulations}

\subsubsection{Instability of {\algB} and tightness of $2/3$ bound.}
We first present an example that shows the tightness of the $2/3$ bound 
on the achievable throughput of \algB{}. Consider a single server where 
jobs have two discrete sizes $0.4$ and $0.6$. The jobs arrive 
according to a Poisson process with average rate $0.014$ jobs per time slot 
and with each job size being equally likely. Each job completes its 
service after a geometric number of time slots with mean $100$.
Observe that by using configuration $(1,1)$ (i.e., 1 spot per 
job type) any arrival rate below $0.02$ jobs per time slot
is supportable. This is not the case though for \algB{} that
schedules based on configurations ${\mathcal{K}}^{(J)}_{RED}$, so it can either schedule two jobs
of size $0.4$ or one job of size $0.6$. This results in 
\algB{} to be unstable for any arrival rate greater than $2/3 \times 0.02 \approx 0.013$.
Both of the other proposed algorithms, \textrm{BF-J/S} and \textrm{VQS-BF}, 
circumvent this problem.
The evolution of the total queue size is depicted in Figure~\ref{fig:simvqunstable}

\begin{figure}
	\centering
	\begin{subfigure}{.5\linewidth}
		\centering
		\includegraphics[width=\linewidth]
		{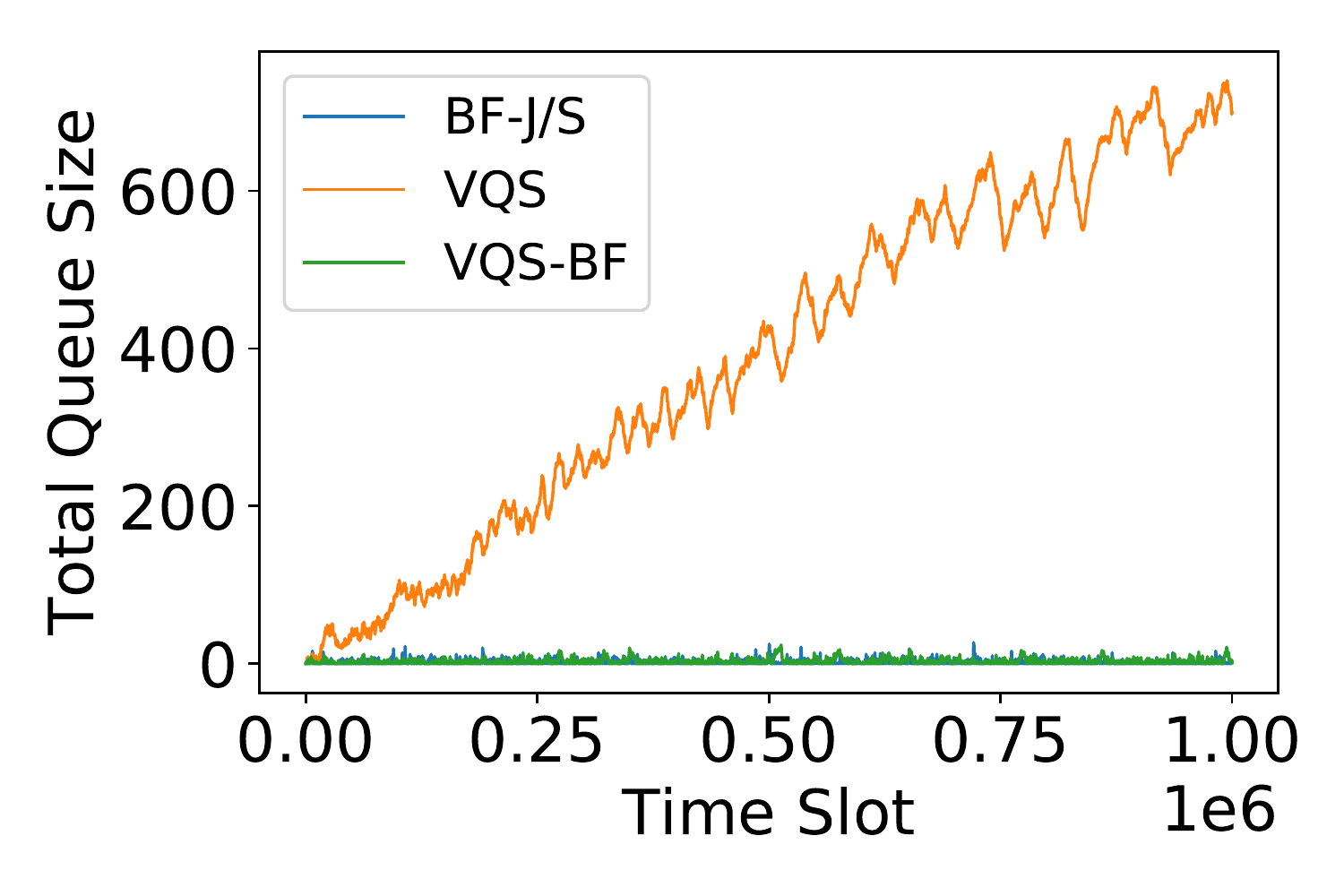}
			\vspace{- 0.3 in}
		\caption{}
		\label{fig:simvqunstable}
	\end{subfigure}%
	\begin{subfigure}{.5\linewidth}
		\centering
		\includegraphics[width=\linewidth]
		{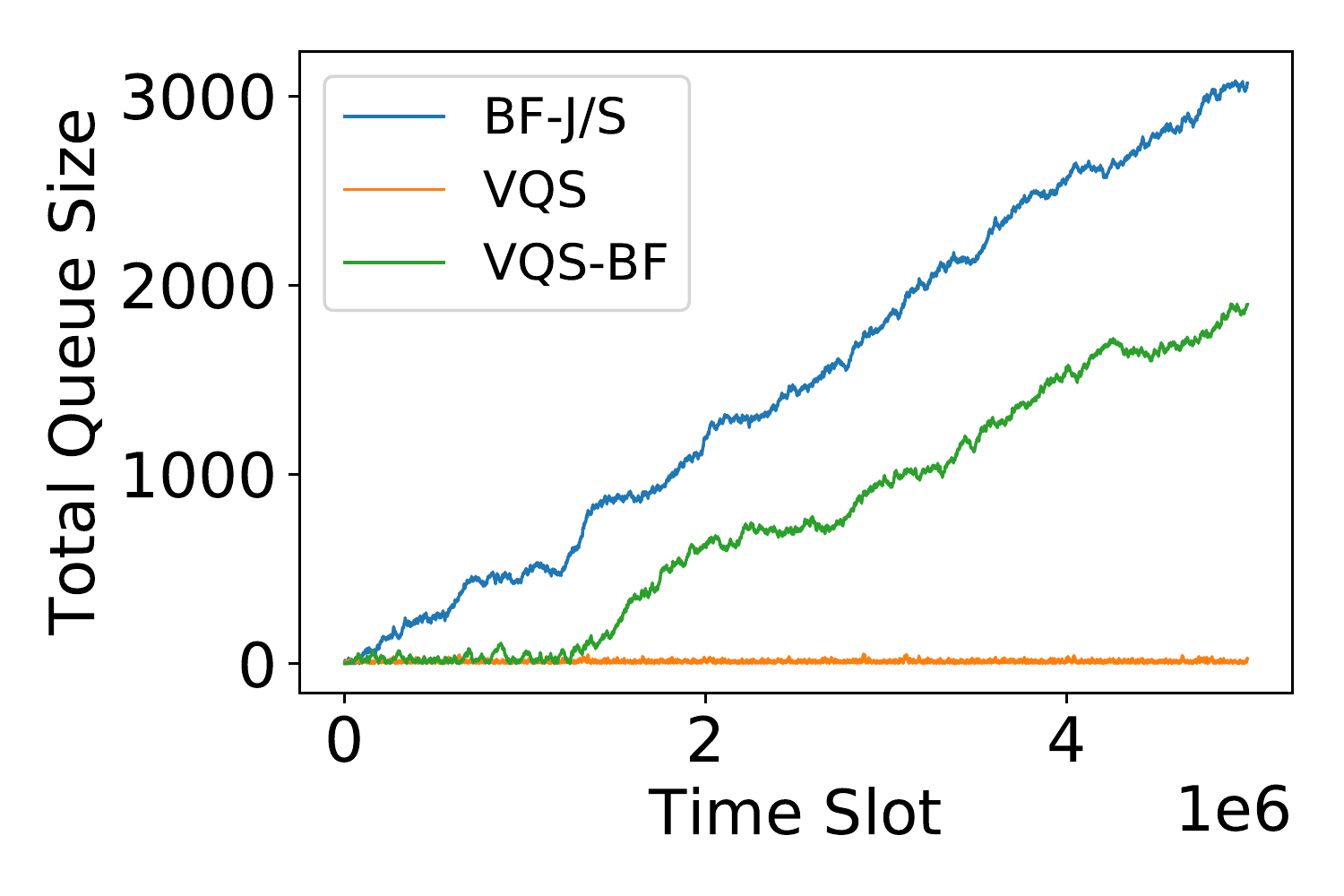}
		\vspace{- 0.3 in}
		\caption{}
		\label{fig:simbfunstable}
	\end{subfigure}%
	\caption{(a) A setting where \textrm{VQS} is unstable, but \textrm{BF} variants are stable. (b) A setting where \textrm{VQS} is stable but \textrm{BF} variants are unstable.}
\end{figure}

\begin{figure*}[t]
	\begin{minipage}{.65\textwidth}
		\begin{subfigure}{0.48\textwidth}
			\includegraphics[width=\textwidth]{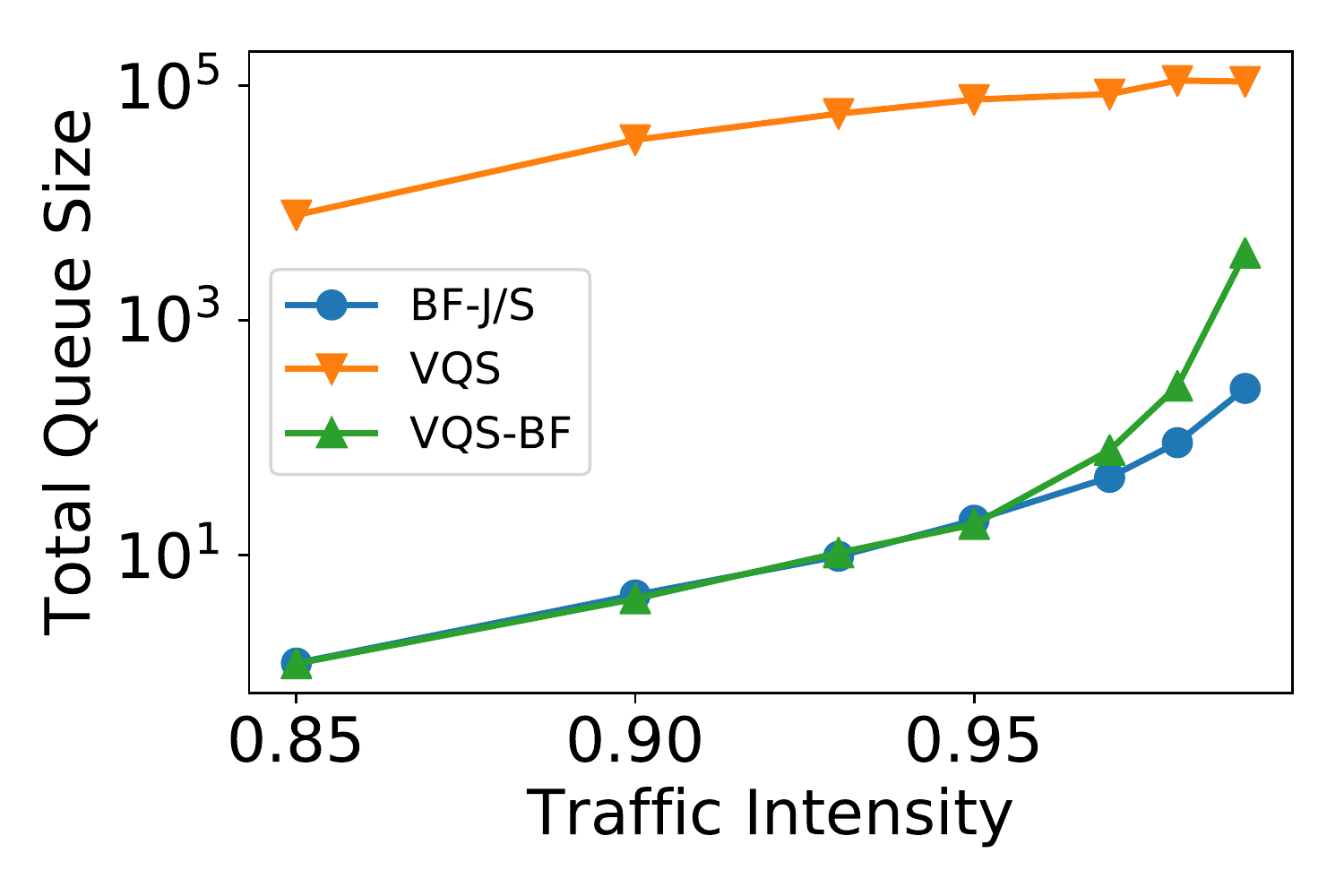}
			\caption{Job sizes$\sim$Unif $[0.01, 0.19]$}
			\label{fig:simunif_1}
		\end{subfigure}%
		\begin{subfigure}{0.48\textwidth}
			\includegraphics[width=\textwidth]{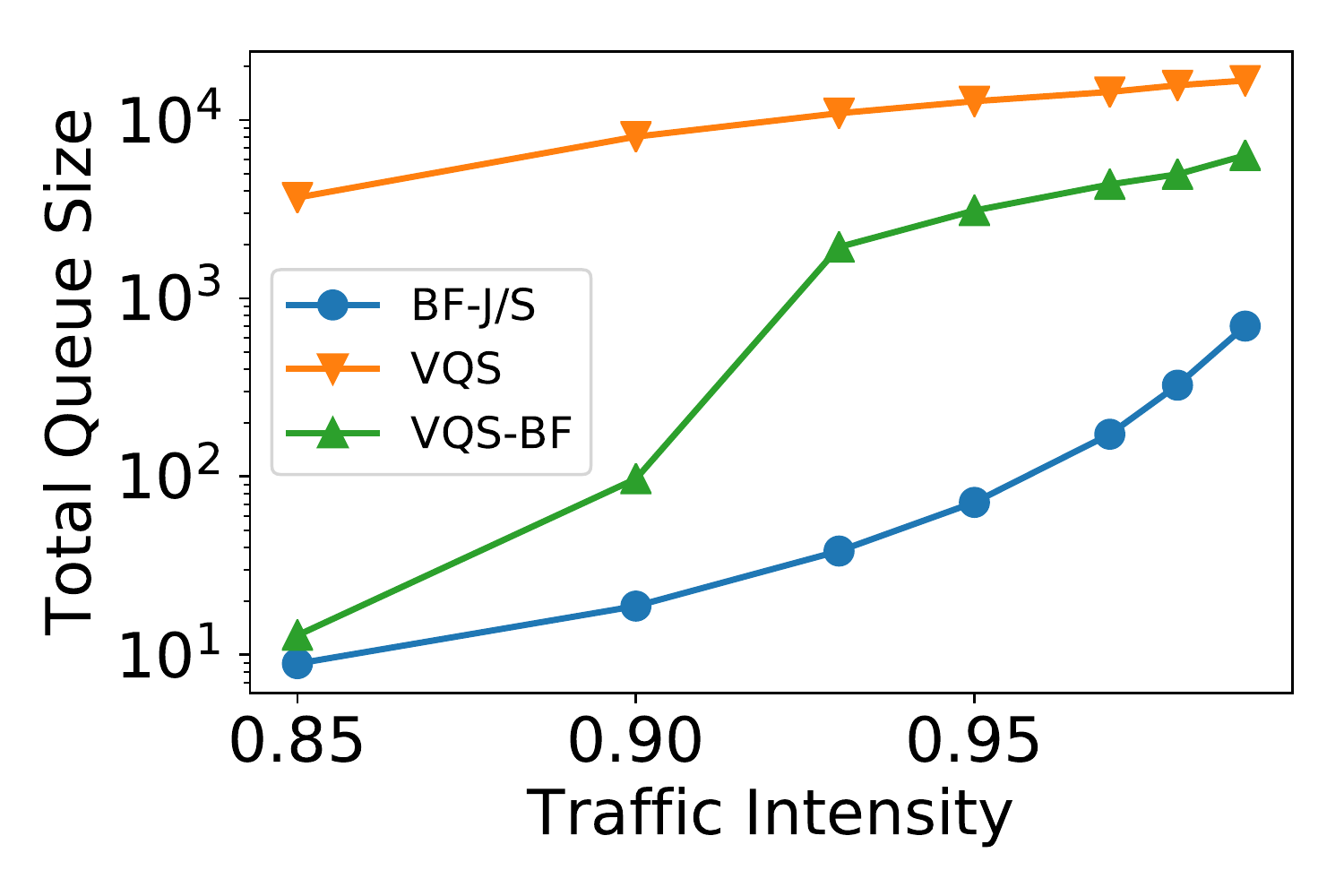}
			\caption{Job sizes$\sim$Unif $[0.1, 0.9]$}
			\label{fig:simunif_2}
		\end{subfigure}
		\caption{Comparison of the average queue size of different algorithms, for various traffic 
			intensities, when job sizes are uniformly distributed in 
			(a) $[0.01, 0.19]$ and (b) $[0.1, 0.9]$, in a system of 
			$5$ servers of capacity 1. }
		\label{fig:simunif}
	\end{minipage}\hfill
	\begin{minipage}{.33\textwidth}
		\includegraphics[width=\textwidth]{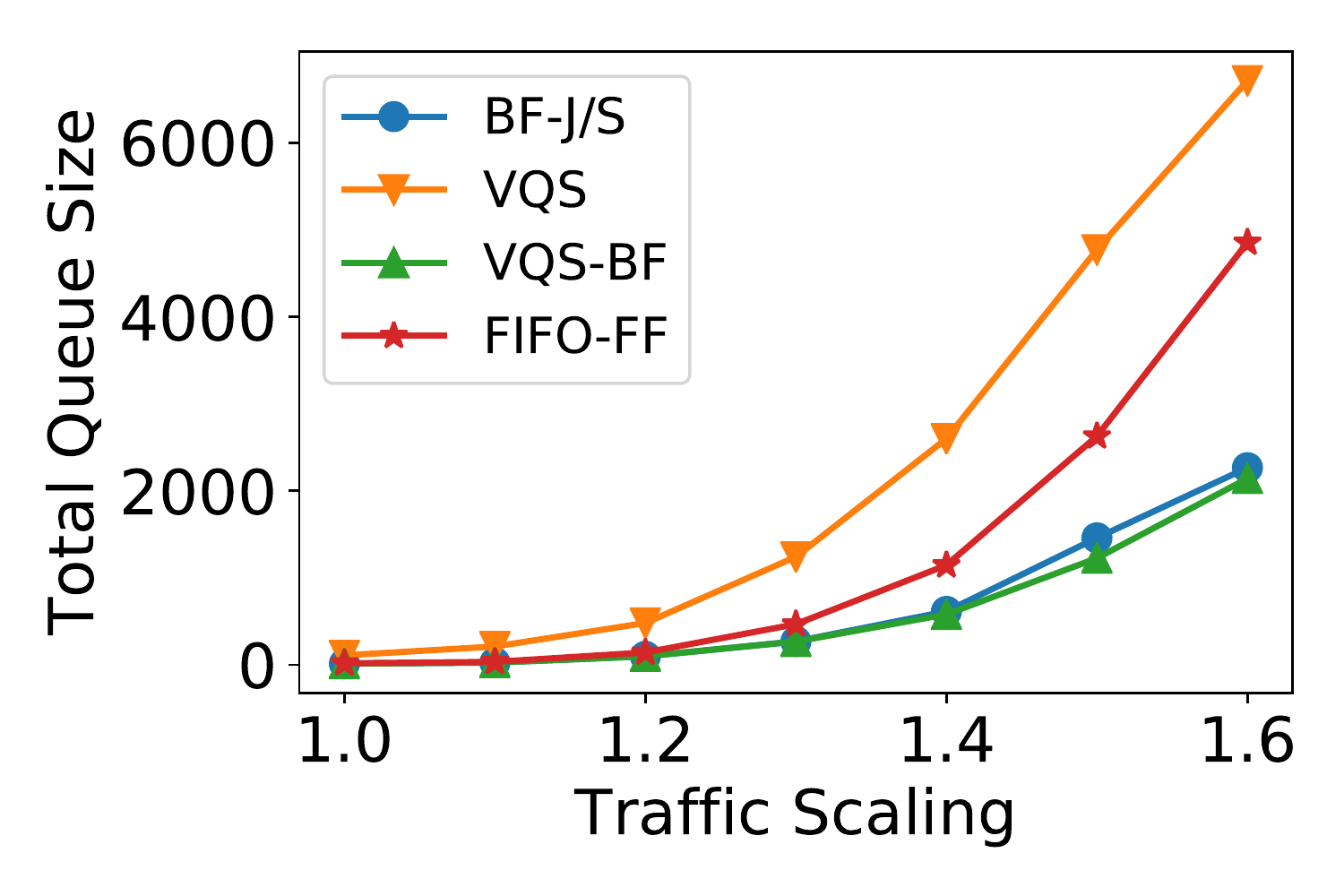}%
		\caption{Comparison of algorithms using Google trace for approximately 
			1,000,000 tasks. Traffic scaling varies from $1$ to $1.6$ and
			number of servers is fixed at $1000$.} 
		\label{fig:simgtracescale}%
	\end{minipage}
\end{figure*}
\subsubsection{Instability of \algA{}-J/S}
We present an example that shows BF-J/S is not stable while VQS can stabilize the queues. Consider a single server of capacity $10$
and that job sizes are sampled from two discrete values
$2$ and $5$. The jobs arrive according to a Poisson process with 
average rate $0.0306$ jobs per time slot, and job of size
$2$ are twice as likely to appear than jobs of size $5$.
Each job completes its service after a fixed number of $100$ time slots .
The evolution of the queue size is depicted in Figure \ref{fig:simbfunstable}.
This shows an example where \algB{} is stable, while both 
\algA{}-J/S and \algB{}-\algA{} are not.

To justify the behavior of the latter two algorithms, we notice that 
under both the server is likely to schedule according to the 
configuration $(2,1)$ that uses two jobs of size $2$ and one of size $5$. 
Because of fixed service times, jobs that are 
scheduled at different time slots, will also depart at different time slots. Hence, it is possible that the scheduling algorithm will not allow 
the configuration $(2,1)$ to change, unless one of the queues empties. However, there is a positive probability that the queues will never get empty since the
expected arrival rate is more than the departure rate for both 
types. The arrival rate vector is $\bm{\lambda} = (0.0204, 0.0102)$ while the departure rate vector $\bm{\mu} = (0.02, 0.01)$.

\algB{} on the other hand will always schedule either five jobs of size $2$
or two of size $5$. The average departure rate in the first configuration is 
$\bm{\mu}_1 = (0.05, 0)$, and in the second configuration  $\bm{\mu}_2 = (0, 0.02)$. The arrival vector is in convex hull of these two vectors as
$\bm{\lambda} < 4/9 \bm{\mu}_1 + 5/9 \bm{\mu}_2$ and therefore is 
supportable. 

\subsubsection{Comparison using Uniform distributions}
To better understand how the algorithms operate under a non-discrete 
distribution of job sizes, we test them using a uniform distribution. We choose $L=5$ servers, each with capacity $1$. We perform two experiments: the job sizes are distributed uniformly over
$[0.01, 0.19]$ in the first experiment and uniformly over $[0.1, 0.9]$ in the second one. Hence $\bar{R}$ is $0.1$ in the first experiment and $0.5$ in the second one.

The service time  of each job is geometrically distributed with mean  $1/\mu=100$ time slots so departure rate is $\mu=0.01$. The job arrivals follow a Poisson process with rate $\mu  L/\bar{R} \times $ jobs per time slot (and thus $\rho=\alpha L/\bar{R} $), where $\alpha$ is a constant which we refer to as ``traffic intensity'' and $L=5$ is the number of servers in these experiments. 
A value of $\alpha=1$ is a bound on what is theoretically supportable 
by any algorithm. 
In each experiment, we change the value of $\alpha$ in the 
interval $[0.85, 0.99]$. The results are depicted in Figure \ref{fig:simunif}.

Overall we can see that \algB{} is worse than other two algorithms in terms of 
average queue size. Algorithms \algA{}-\textrm{J/S} and \algB{}-\algA{} 
look comparable in the first experiment for traffic intensities up to $0.95$,
otherwise \algA{}-\textrm{J/S} has a clear advantage.
An interpretation of results is that
\algB{} and \algB{}-\algA{} have particularly worse delays when the average job size 
is large, since large jobs cannot be scheduled most of the time, unless they 
are part of the active configuration of a server. That makes these 
algorithm less flexible compared to \algA{}-\textrm{J/S} for 
scheduling such jobs.

\subsection{Google Trace Simulations}
We test the algorithms using a traffic trace from a Google cluster dataset~\cite{ClusterData}. We performed the following preprocessing on the dataset:
\begin{itemize}[leftmargin=3mm]
\item We filtered the tasks and kept those that were completed without interruptions/errors. 
\item All tasks had two resources, CPU and memory. To convert them to a
	single resource, we used the maximum of the two requirements which were 
    already normalized in $[0,1]$ scale. 
\item The servers had two resources, 
CPU and memory, and change over time as they a updated or replaced. For simplicity, we consider a fixed number of servers, each with a single resource capacity normalized to $1$.
\item Trace events are in microsec accuracy. In our algorithms,
we make scheduling decisions every $100$ msec.
\item We used a part of the trace corresponding to about a million task arrivals 
spanning over approximately 1.5 days.
\end{itemize}

We compare the algorithms proposed in this work and a baseline based on Hadoop's default 
FIFO scheduler~\cite{hadoop}. While the original FIFO scheduler is slot-based~\cite{usama2017job}, 
the FIFO scheduler considered here schedules jobs in a FIFO manner, by attempting to 
pack the first job in the queue to the first server that has sufficient capacity to 
accommodate the job. We refer to this scheme as \textrm{FIFO-FF} which should perform better than the slot-based FIFO, since it packs jobs in servers (using First-Fit) instead of using predetermined slots.

We scale the job arrival rate by multiplying the arrival times of tasks 
by a factor $\beta$. We refer to $1/\beta$ as ``traffic scaling'' 
because larger $1/\beta$ implies that more jobs arrive in a time unit. 
The number of servers was fixed to $1000$, while traffic scaling varied from $1$ to $1.6$.
The average queue sizes are depicted in Figure~\ref{fig:simgtracescale}.
As traffic scaling increases, \algA{}-\textrm{J/S} and \algB{}-\algA{} have a clear advantage over the other schemes, with \algB{}-\algA{} also yielding a small improvement in the queue size compared to \algA{}-\textrm{J/S}.
It is interesting that \algB{}-\algA{} has a consistent advantage over \algA{}-\textrm{J/S} at higher
traffic, albeit small, although both algorithms are greedy in the way that they pack jobs in  servers.

\section{Discussion and Open Problems}\label{sec:dis}
In this work, we designed three scheduling algorithms for jobs whose 
sizes come from a general unknown distribution. 
Our algorithms achieved two goals:  keeping the complexity low, and providing throughput guarantees for any 
distribution of job sizes, \textit{without} actually knowing the prior distribution.

Our results, however, are lower bounds on the performance of the algorithms and 
simulation results show that the algorithms {\algA}-\textrm{J/S} and {\algB}-{\algA} may support workloads that go beyond 
their theoretical lower bounds. It remains as an open problem to tighten the 
lower bounds or construct upper bounds that approach the lower bounds. 

In addition, we made some simplifying assumptions in our model but results indeed hold under more general models. 
One of the assumptions was that the servers are homogeneous.
{\algA}-\textrm{J/S} and our analysis can indeed be easily applied without this assumption.
For \algB{}
and \algB{}-\algA{}, the scheduling can be also applied without changes
when servers have resources that differ by a power of $2$ which is a
common case. As a different approach, we can maintain different sets 
of virtual queues, one set for each type of servers.

Another assumption was that service durations follow geometric distribution. This assumption was made to simplify the proofs, as it justifies that a server will empty in a finite expected time by chance. Since this may not happen under general service time distributions (e.g. one may construct adversarial service durations that prevent server from becoming empty), in all our algorithms we can incorporate  a stalling technique proposed in~\cite{PG17} that actively forces a server to become empty by preventing it from scheduling new jobs. 
The decision to stall a server is made whenever server operates in an ``inefficient'' configuration. For \algA{}-\textrm{J/S} that condition is when the server is less
than half full, while for \algB{} and \algB{}-\algA{}, is when the weight of 
configuration of a server is far from the maximum weight over ${\mathcal{K}}^{(J)}_{RED}$.

Finally we based our scheduling decisions on a single resource. Depending
on workload, this may cause different levels of fragmentation, but resource 
requirements will not be violated if resources of jobs are mapped to the maximum 
resource (e.g. like our preprocessing on Google trace data). 
A more efficient approach is to extend {\algA}-\textrm{J/S} to multi-resource setting, 
by considering a Best-Fit score as a linear combination of per-resource occupancies. It has been empirically 
shown in ~\cite{grandl2015multi} that the inner product of the vector of the job's 
resource requirements and the vector of server's occupied resources is a good candidate. 
We leave the theoretical study of scheduling jobs with multi-resource distribution 
as a future research.

\begingroup
\setlength{\emergencystretch}{8em}
\bibliographystyle{IEEEtran}
\bibliography{javad}

\begin{thebibliography}{10}
\providecommand{\url}[1]{#1}
\csname url@samestyle\endcsname
\providecommand{\newblock}{\relax}
\providecommand{\bibinfo}[2]{#2}
\providecommand{\BIBentrySTDinterwordspacing}{\spaceskip=0pt\relax}
\providecommand{\BIBentryALTinterwordstretchfactor}{4}
\providecommand{\BIBentryALTinterwordspacing}{\spaceskip=\fontdimen2\font plus
\BIBentryALTinterwordstretchfactor\fontdimen3\font minus
  \fontdimen4\font\relax}
\providecommand{\BIBforeignlanguage}[2]{{%
\expandafter\ifx\csname l@#1\endcsname\relax
\typeout{** WARNING: IEEEtran.bst: No hyphenation pattern has been}%
\typeout{** loaded for the language `#1'. Using the pattern for}%
\typeout{** the default language instead.}%
\else
\language=\csname l@#1\endcsname
\fi
#2}}
\providecommand{\BIBdecl}{\relax}
\BIBdecl

\bibitem{hadoop}
``{Apache Hadoop},'' \url{https://hadoop.apache.org}, 2018.

\bibitem{spark}
``{Apache Spark},'' \url{https://spark.apache.org}, 2018.

\bibitem{hive}
``{Apache Hive},'' \url{https://hive.apache.org}, 2018.

\bibitem{ClusterData}
J.~Wilkes, ``{Google Cluster Data},''
  \url{https://github.com/google/cluster-data}, 2011.

\bibitem{Martello1990}
S.~Martello and P.~Toth, \emph{Knapsack Problems: Algorithms and Computer
  Implementations}.\hskip 1em plus 0.5em minus 0.4em\relax New York, NY, USA:
  John Wiley \& Sons, Inc., 1990.

\bibitem{magsriyi12}
S.~T. Maguluri, R.~Srikant, and L.~Ying, ``Stochastic models of load balancing
  and scheduling in cloud computing clusters,'' in \emph{Proceedings of IEEE
  INFOCOM}, 2012, pp. 702--710.

\bibitem{Stolyar2013}
A.~L. Stolyar, ``An infinite server system with general packing constraints,''
  \emph{Operations Research}, vol.~61, no.~5, pp. 1200--1217, 2013.

\bibitem{Siva13}
S.~T. Maguluri and R.~Srikant, ``Scheduling jobs with unknown duration in
  clouds,'' in \emph{Proceedings 2013 IEEE INFOCOM}, 2013, pp. 1887--1895.

\bibitem{Siva14}
------, ``Scheduling jobs with unknown duration in clouds,'' \emph{IEEE/ACM
  Transactions on Networking}, vol.~22, no.~6, pp. 1938--1951, 2014.

\bibitem{ghaderi2016randomized}
J.~Ghaderi, ``Randomized algorithms for scheduling {VM}s in the cloud,'' in
  \emph{IEEE INFOCOM}, 2016, pp. 1--9.

\bibitem{PG17}
K.~Psychas and J.~Ghaderi, ``On non-preemptive {VM} scheduling in the cloud,''
  \emph{Proc. ACM Meas. Anal. Comput. Syst. (ACM SIGMETRICS 2018)}, vol.~1,
  no.~2, pp. 35:1--35:29, Dec. 2017.

\bibitem{isard2009quincy}
M.~Isard, V.~Prabhakaran, J.~Currey, U.~Wieder, K.~Talwar, and A.~Goldberg,
  ``Quincy: fair scheduling for distributed computing clusters,'' in
  \emph{Proc. of the ACM SIGOPS symposium on operating systems principles},
  2009, pp. 261--276.

\bibitem{tang2013dynamic}
S.~Tang, B.-S. Lee, and B.~He, ``Dynamic slot allocation technique for
  mapreduce clusters,'' in \emph{IEEE International Conference on Cluster
  Computing (CLUSTER)}, 2013, pp. 1--8.

\bibitem{grandl2015multi}
R.~Grandl, G.~Ananthanarayanan, S.~Kandula, S.~Rao, and A.~Akella,
  ``Multi-resource packing for cluster schedulers,'' \emph{ACM SIGCOMM Computer
  Communication Review}, vol.~44, no.~4, pp. 455--466, 2015.

\bibitem{Verma2015}
A.~Verma, L.~Pedrosa, M.~Korupolu, D.~Oppenheimer, E.~Tune, and J.~Wilkes,
  ``{Large-scale cluster management at Google with Borg},'' \emph{European
  Conference on Computer Systems - EuroSys}, pp. 1--17, 2015.

\bibitem{Ghodsi2011}
A.~Ghodsi, M.~Zaharia, B.~Hindman, A.~Konwinski, S.~Shenker, and I.~Stoica,
  ``Dominant resource fairness: Fair allocation of multiple resource types,''
  \emph{NSDI}, vol. 167, no.~1, pp. 24--24, 2011.

\bibitem{chowdhury2016hug}
M.~Chowdhury, Z.~Liu, A.~Ghodsi, and I.~Stoica, ``Hug: Multi-resource fairness
  for correlated and elastic demands.'' in \emph{NSDI}, 2016, pp. 407--424.

\bibitem{usama2017job}
M.~Usama, M.~Liu, and M.~Chen, ``Job schedulers for big data processing in
  hadoop environment: testing real-life schedulers using benchmark programs,''
  \emph{Digital Communications and Networks}, vol.~3, no.~4, pp. 260--273,
  2017.

\bibitem{Fairscheduler}
``{Hadoop: Fair Scheduler},''
  \url{https://hadoop.apache.org/docs/current/hadoop-yarn/hadoop-yarn-site/FairScheduler.html},
  2018.

\bibitem{Capscheduler}
``{Hadoop: Capacity Scheduler},'' \url{https://hadoop.apache.org/docs/current/
  hadoop-yarn/hadoop-yarn-site/CapacityScheduler.html}, 2018.

\bibitem{tassiulas1992stability}
L.~Tassiulas and A.~Ephremides, ``Stability properties of constrained queueing
  systems and scheduling policies for maximum throughput in multihop radio
  networks,'' \emph{IEEE Transactions on Automatic Control}, vol.~37, no.~12,
  pp. 1936--1948, 1992.

\bibitem{garey2}
M.~R~Garey and D.~S~Johnson, ``Computers and intractability: A guide to the
  theory of {NP}-completeness,'' \emph{WH Freeman \& Co.}, 1979.

\bibitem{john}
D.~S. Johnson, A.~Demers, J.~D. Ullman, M.~R. Garey, and R.~L. Graham,
  ``Worst-case performance bounds for simple one-dimensional packing
  algorithms,'' \emph{SIAM Journal on Computing}, vol.~3, no.~4, pp. 299--325,
  1974.

\bibitem{kenyon2002linear}
C.~Kenyon and M.~Mitzenmacher, ``Linear waste of best fit bin packing on skewed
  distributions,'' \emph{Random Structures \& Algorithms}, vol.~20, no.~3, pp.
  441--464, 2002.

\bibitem{coffman1996approximation}
E.~G. Coffman~Jr, M.~R. Garey, and D.~S. Johnson, ``Approximation algorithms
  for bin packing: A survey,'' in \emph{Approximation algorithms for NP-hard
  problems}.\hskip 1em plus 0.5em minus 0.4em\relax PWS Publishing Co., 1996,
  pp. 46--93.

\bibitem{Shah2008}
D.~Shah and J.~N. Tsitsiklis, ``{Bin packing with queues},'' \emph{Journal of
  Applied Probability}, vol.~45, no.~4, pp. 922--939, 2008.

\bibitem{coffman2001bandwidth}
E.~Coffman and A.~L. Stolyar, ``Bandwidth packing,'' \emph{Algorithmica},
  vol.~29, no. 1-2, pp. 70--88, 2001.

\bibitem{gamarnik2004stochastic}
D.~Gamarnik, ``Stochastic bandwidth packing process: stability conditions via
  lyapunov function technique,'' \emph{Queueing systems}, vol.~48, no. 3-4, pp.
  339--363, 2004.

\bibitem{nitu2018working}
V.~Nitu, A.~Kocharyan, H.~Yaya, A.~Tchana, D.~Hagimont, and H.~Astsatryan,
  ``Working set size estimation techniques in virtualized environments: One
  size does not fit all,'' \emph{Proc. of the ACM on Measurement and Analysis
  of Computing Systems}, vol.~2, no.~1, p.~19, 2018.

\bibitem{Foss2004}
S.~Foss and T.~Konstantopoulos, ``{An overview of some stochastic stability
  methods},'' \emph{Journal of the Operations Research Society of Japan},
  vol.~47, no.~4, pp. 275--303, 2004.

\bibitem{Billingsley1999}
P.~Billingsley, \emph{{Convergence of Probability Measures 2e}}.\hskip 1em plus
  0.5em minus 0.4em\relax John Wiley {\&} Sons, Inc, 1999.

\bibitem{ethier2009markov}
\BIBentryALTinterwordspacing
S.~Ethier and T.~Kurtz, \emph{Markov Processes: Characterization and
  Convergence}, ser. Wiley Series in Probability and Statistics.\hskip 1em plus
  0.5em minus 0.4em\relax Wiley, 2009. [Online]. Available:
  \url{https://books.google.com/books?id=zvE9RFouKoMC}
\BIBentrySTDinterwordspacing

\bibitem{tweedie1976criteria}
R.~Tweedie, ``Criteria for classifying general markov chains,'' \emph{Advances
  in Applied Probability}, vol.~8, no.~4, pp. 737--771, 1976.

\end{thebibliography}
\endgroup

\appendix[Proofs]
\subsection{Proof of Theorem~\ref{thm:OPT}}\label{pr:thm_OPT}
We first prove the theorem for continuous probability distributions, then show how to handle discontinuities in general distributions.  

Define partition ${X}^{(n)}$ to be the 
collection of $m_n=2^{n+1}$ intervals $X_i^{(n)}:(\xi_{i-1}^{(n)}, \xi_{i}^{(n)}]$, 
such that $\xi_0^{(n)}=0$, $\xi_{2^{n+1}}^{(n)}=1$, and $F_R(\xi_i^{(n)})=\frac{i}{2^{n+1}}$, for $i=1, \cdots, 2^{n+1}-1$. This construction is possible since $F_R$ is an increasing continuous function, 
hence $F_R(x)=c$ always has a unique solution $x \in [0,1]$ if $c \in [0,1]$. Subsequently, 
$$
\pi_i:=\mathds{P}\left(R \in I_i^{(n)}\right) = \frac{1}{2^{n+1}},\ i=1, \cdots, m_n.
$$

In the rest of the proof, we use the following notations.
$\mathbf{1}_{N}$ is a vector of all ones of length $N$.
$\mathbf{e}_{i}$ is a basis vector with value $1$ in its $i$th entry and
$0$ elsewhere. $\rho^\star$ is the
maximum workload that can be supported by any algorithm with the given 
distribution of job sizes $F_R$. Also $\overline{\rho}^\star({X}^{(n)})$ is the
maximum supportable workload when 
upper-rounded queues are used under partition ${X}^{(n)}$ and 
$\underline{\rho}^\star({X}^{(n)})$ the respective maximum workload,
when lower-rounded queues are used.

Under upper-rounded or lower-rounded virtual queues, 
job sizes have $2^{n+1}$ discrete values, 
which makes the problem equivalent with scheduling $2^{n+1}$ job types.
For notational purpose, we define the \textit{workload vector} $\bm{\rho}=\rho \bm{\pi}$ 
where $\bm{\pi}=\left(\pi_i, i = 1, \cdots, m_n \right)$ is the vector of probabilities of the types and $\rho$ is the workload of the system. Hence, 
under upper-rounded queues, the workload vector is 
$\bm{\rho}_1 = \frac{\overline{\rho}({X}^{(n)})}{2^{n+1}} \mathbf{1}_{m_n}$.

Using lower-rounded virtual queues is equivalent to using upper-rounded
virtual queues, but the workload vector is instead
$\bm{\rho}_2 = \frac{\underline{\rho}({X}^{(n)})}{2^{n+1}}
(\mathbf{1}_{m_n} - \mathbf{e}_{m_n})$. This is because we can 
essentially ignore the jobs whose sizes are rounded to 0 and no job size can 
be rounded to $1$. 

With discrete job types whose sizes are $\xi^{(n)}_{i}$ for $i=1, \ldots m_n$, 
we can extend the notion of \textit{feasible configuration} in 
Definition~\ref{def:conf} to jobs of a continuous distribution. 
In this case, configuration $\mathbf{k}$ is a $m_n$-dimensional vector
and the set of feasible configurations is denoted by $\mathcal{\bar{K}}$. 
The workload is supportable if it is in the convex hull of set of 
feasible configurations, as in \dref{eq:region}.
With the upper-rounded queues, and given all $L$ servers are the same
and have the same set of feasible configurations, there should exist $p_{\mathbf{k}}\geq 0$, $\mathbf{k} \in \overline{\mathcal{K}}$, such that  
\begin{equation} \label{eq:rho1}
L\sum_{\mathbf{k} \in \overline{\mathcal{K}}} p_{\mathbf{k}}
\mathbf{k} > \bm{\rho}_1, \ \ 
\sum_{\mathbf{k} \in \overline{\mathcal{K}}} p_{\mathbf{k}} = 1.
\end{equation}
Similarly, with lower-rounded virtual queues, there should exist $q_{\mathbf{k}}\geq 0$, $\mathbf{k} \in \overline{\mathcal{K}}$, such that  
\begin{equation} \label{eq:rho2}
L\sum_{\mathbf{k} \in \overline{\mathcal{K}}\setminus
	\{\mathbf{e}_{m_n}\}} q_{\mathbf{k}}
\mathbf{k} > \bm{\rho}_2, \ \ 
\sum_{\mathbf{k} \in \overline{\mathcal{K}}\setminus
	\{\mathbf{e}_{m_n}\}} q_{\mathbf{k}} = 1.
\end{equation}

Jobs of size $1$ can be served only by configuration $\mathbf{e}_{m_n}$,
i.e., server is filled with a single job of size $1$. Hence we can split the 
first equation of (\ref{eq:rho1}) into
\begin{equation}\label{eq:rhoover}
\begin{aligned}
&L \sum_{\mathbf{k} \in \overline{\mathcal{K}}\setminus
	\{\mathbf{e}_{m_n}\}} p_{\mathbf{k}} \mathbf{k} >
\frac{\overline{\rho}({X}^{(n)})}{2^{n+1}} \left(\mathbf{1}_{m_n} -
\mathbf{e}_{m_n}\right), \\
&L p_{\mathbf{e}_{m_n}} \mathbf{e}_{m_n} >
\frac{\overline{\rho}({X}^{(n)})}{2^{n+1}} \mathbf{e}_{m_n}. \\
\end{aligned}
\end{equation}
Also given that with lower-rounded virtual queues there are no jobs of 
size $1$, the first equation of (\ref{eq:rho2}) becomes:
\begin{equation}\label{eq:rhounder}
\begin{aligned}
L \sum_{\mathbf{k} \in \overline{\mathcal{K}}\setminus
	\{\mathbf{e}_{m_n}\}} q_{\mathbf{k}}
\mathbf{k} > \frac{\underline{\rho}({X}^{(n)})}{2^{n+1}}
(\mathbf{1}_{m_n} - \mathbf{e}_{m_n}). \\
\end{aligned}
\end{equation}

Now if we replace $\overline{\rho}({X}^{(n)})$ with $\overline{\rho}^\star({X}^{(n)})$
and $\underline{\rho}({X}^{(n)})$ with $\underline{\rho}^\star({X}^{(n)})$,
inequalities in (\ref{eq:rhoover}) and (\ref{eq:rhounder})
must hold with equality by definition, i.e., 
\begin{equation}\label{eq:help2}
\begin{aligned}
&L \sum_{\mathbf{k} \in \overline{\mathcal{K}} \setminus
	\{\mathbf{e}_{m_n}\}} p_{\mathbf{k}} \mathbf{k} =
\frac{\overline{\rho}^\star({X}^{(n)})}{2^{n+1}} \left(\mathbf{1}_{m_n} -
\mathbf{e}_{m_n}\right) \\
& L p_{\mathbf{e}_{m_n}} \mathbf{e}_{m_n} =
\frac{\overline{\rho}^\star({X}^{(n)})}{2^{n+1}} \mathbf{e}_{m_n}. \\
&L \sum_{\mathbf{k} \in \overline{\mathcal{K}} \setminus
	\{\mathbf{e}_{m_n}\}} q_{\mathbf{k}}
\mathbf{k} = \frac{\underline{\rho}^\star({X}^{(n)})}{2^{n+1}}
(\mathbf{1}_{m_n} - \mathbf{e}_{m_n}).
\end{aligned}
\end{equation}
Notice that the direction of 
vectors $\sum_{\mathbf{k} \in \overline{\mathcal{K}}\setminus
	\{\mathbf{e}_{m_n}\}} p_{\mathbf{k}}\mathbf{k}$ 
and $\sum_{\mathbf{k} \in \overline{\mathcal{K}}\setminus
	\{\mathbf{e}_{m_n}\}} q_{\mathbf{k}}\mathbf{k}$
is the same. Given a solution $p_{\mathbf{k}}$, $\mathbf{k} \in \overline{\mathcal{K}}$, to \dref{eq:help2}, it is sufficient to choose $q_{\mathbf{k}}$ to be proportional to $p_{\mathbf{k}}$.
Assuming $p_{\mathbf{k}}$ and $q_{\mathbf{k}}$ are proportional, and noting that by definition,  
\begin{equation}
\begin{aligned}
 \sum_{\mathbf{k} \in \overline{\mathcal{K}}\setminus
	\{\mathbf{e}_{m_n}\}} p_{\mathbf{k}} = 1 - p_{\mathbf{e}_{m_n}}, \ 
\sum_{\mathbf{k} \in \overline{\mathcal{K}}\setminus
	\{\mathbf{e}_{m_n}\}} q_{\mathbf{k}} = 1,
\end{aligned}
\end{equation}
it should hold that $p_{\mathbf{k}} = \left(1 - p_{\mathbf{e}_{m_n}} \right) q_{\mathbf{k}}$
and hence $\overline{\rho}^\star({X}^{(n)}) = \left(1-p_{\mathbf{e}_{m_n}} \right)
\underline{\rho}^\star({X}^{(n)})$. From the second equation of (\ref{eq:rhoover})
we get $\overline{\rho}^\star({X}^{(n)}) = 2^{n+1} L p_{\mathbf{e}_{m_n}}$.
Using these two equations, we can write $\underline{\rho}^\star({X}^{(n)})$
as a function of $\overline{\rho}^\star({X}^{(n)})$, 
i.e.,
\begin{equation}
\underline{\rho}^\star({X}^{(n)}) = 
\frac{\overline{\rho}^\star({X}^{(n)})}{1 - 
	\frac{\overline{\rho}^\star({X}^{(n)})}{L 2^{n+1}}} 
\end{equation}
which implies
\begin{equation}\label{eq:overunder}
\underline{\rho}^\star({X}^{(n)}) - 
\overline{\rho}^\star({X}^{(n)}) =
\frac{\underline{\rho}^\star({X}^{(n)})^2}{L 2^{n+1}-
	\underline{\rho}^\star({X}^{(n)})}.
\end{equation}

By construction, $\underline{\rho}^\star({X}^{(n)})$ is a decreasing
sequence in $i$, so it is bounded from above by $\underline{\rho}^\star({X}^{(0)})$ 
and from below by $0$. Similarly $\overline{\rho}^\star({X}^{(n)})$ is an 
increasing sequence with the same bounds. By the monotone convergence theorem, 
the limits of both exist and by construction 
$\underline{\rho}^\star({X}^{(n)}) - 
\overline{\rho}^\star({X}^{(n)}) > 0$.
Then assuming $n$ is large enough so that
$2^{n+1} L > \underline{\rho}^\star({X}^{(0)})$, 
\begin{equation*}
0 \le \lim_{n \to \infty} \underline{\rho}^\star({X}^{(n)}) -
\overline{\rho}^\star({X}^{(n)})  \le \lim_{i \to \infty}
\frac{\underline{\rho}^\star({X}^{(0)})^2}{L 2^{n+1}-
	\underline{\rho}^\star({X}^{(0)})} = 0
\end{equation*}
and
$\lim_{n \to \infty} \overline{\rho}^\star({X}^{(n)}) = 
\overline{\rho}^\star=
\lim_{n \to \infty} \underline{\rho}^\star({X}^{(n)}) = 
\underline{\rho}^\star = \rho^\star.$

\textit{General Probability Distribution:} 
The proof for a general probability distribution follows 
similar arguments but the sequence of partitions has to
change to include points of discontinuity. Specifically, let the points of discontinuity be $x_j$ and their probability 
$P_j \equiv \mathds{P}(r=x_j)$, $0 < x_j \le 1$, $ j \in \mathds{N}$ 
(since $F_R$ is monotone, we know that the number of discontinuities is countable).
Define the partial sum of probabilities $ S_N \equiv \sum_{j=0}^{N} P_j$.
Sequence of partial sums is certainly convergent, as it is bounded above by $1$ 
and is increasing, so let its limit be
$
\lim_{N \to \infty} S_N = P.
$
By the convergence property, there exists $M_n$ such that
\begin{equation}
| S_k - P | < \frac{1}{2^{n+1}} \quad \forall k \ge M_n.
\end{equation}
Following this, we can define the continuous part of $F_R$ to be
\begin{equation}
F_R^{(c)}(x) \equiv \mathds{P}(r \le x) - 
\sum_{j \in \mathds{N}} P_j \mathds{1}(x_j \le x).
\end{equation}
By definition we have $F_R^{(c)}(1) = 1 - P$. As a result, we can define, similarly to the proof of 
continuous case,  
the $2^{n+1}$ intervals $X_i^{(n)} : (\xi_{i-1}^{(n)}, \xi_{i}^{(n)}]$ with 
$\xi_0^{(n)}=0$, $\xi_{2^{n+1}}^{(n)}=1$, and 
$F_R^{(c)}(\xi_i^{(n)})=\frac{i (1 - P)}{2^{n+1}}$ for $i=1, \ldots, 2^{n+1}-1$.
Compared to the proof of continuous case though, we need to change the 
partition ${X}^{(n)}$ to be the following $2^{n+1} + M_n + 2$ 
sets
\begin{equation}
\begin{aligned}
&\left\{x_i\right\}, \quad i = 0, \cdots M_n, \\
&\left\{x_i: i > M_n\right\}, \\
&(\xi_{i-1}^{(n)}, \xi_{i}^{(n)}]\setminus \left\{x_k: k \in \mathds{N}\right\},
\quad i = 1, \cdots 2^{n+1}, \\
\end{aligned}
\end{equation}

The virtual queue corresponding to set $\left\{x_i: i > M_n\right\}$ is
different from the rest in the way that rounding is done when working with upper-rounded or lower-rounded 
virtual queues. 
In the former case, we round up its jobs to $1$, and in the latter we round down its jobs to $0$. While this diverges from Definition~\ref{def:rounded}, it is convenient to round the job sizes in this special queue to $1$ and $0$ rather than to $\sup$ and $\inf$.

Next we use symbol $\Vert$ to describe concatenation of
two vectors, e.g. if $\mathbf{x} = (x_1, \cdots x_M)$ and 
$\mathbf{y} = (y_1, \cdots y_N)$ then $(\mathbf{x} \Vert \mathbf{y}) = 
(x_1, \cdots x_M, y_1, \cdots y_N)$.

The configurations will now have $m_n = 2^{n+1} + M_n + 1$ types.
When upper-rounded virtual queues are used, the workload vector is
\begin{equation*}
\begin{aligned}
&\bm{\rho}_1 = \overline{\rho}\left({X}^{(n)}\right) \times \\
&\left(
(P_0, \cdots, P_{M_n}) \Vert
(1 - P)/2^{n+1} \mathbf{1}_{2^{n+1}} + | S_{M_n} - P | \mathbf{e}_{2^{n+1}}
\right)
\end{aligned}
\end{equation*}
and when lower-rounded virtual queues are used, 
\begin{equation*}
\bm{\rho}_2 = \underline{\rho}\left({X}^{(n)}\right) \left(
(P_0, \cdots, P_{M_n}) \Vert
(1 - P)/2^{n+1} \left(\mathbf{1}_{2^{n+1}} - \mathbf{e}_{2^{n+1}}
\right)\right)
\end{equation*}
The vectors $\bm{\rho}_1$ and $\bm{\rho}_2$ have the same direction if one ignores the last index that
corresponds to job types of size $1$. If $m_n$ is that index, with similar 
arguments as in the proof of continuous case, we can conclude 
 $\overline{\rho}\left({X}^{(n)}\right) = 
(1 - p_{\mathbf{e}_{m_n}}) \underline{\rho}\left({X}^{(n)}\right)$
and
$\overline{\rho}\left({X}^{(n)}\right) = 
\frac{L p_{\mathbf{e}_{m_n}}}{\frac{1-P}{2^{n+1}} + (P-S_{M_n})} \ge
2^n L p_{\mathbf{e}_{m_n}}$, and, equivalently to~\dref{eq:overunder},
\begin{equation}
\begin{aligned}
&\underline{\rho}^\star({X}^{(n)}) - 
\overline{\rho}^\star({X}^{(n)}) \le
\frac{\underline{\rho}^\star({X}^{(n)})^2}{L 2^n-
	\underline{\rho}^\star({X}^{(n)})}.
\end{aligned}
\end{equation}

The rest of the arguments is the same as in
the continuous distribution case.

\subsection{Proof of Theorem~\ref{thm:GF1}}\label{pr:thm_GF1}
The \textit{state} of our system at time slot $t$ is
\begin{equation}
\state{S}(t) = (\mathcal{Q}(t), \bm{\mathcal{H}}(t)).
\end{equation}

Recall that $\mathcal{Q}(t)$ is the set of jobs in queue and its cardinality
is $|\mathcal{Q}(t)| = Q(t)$, and
$\bm{\mathcal{H}}(t)=(\mathcal{H}_\ell(t), \ell \in \mathcal{L})$ is
the set of scheduled jobs in servers $\mathcal{L}$.
We will denote the set of all feasible states as $\mathcal{S}$.

An equivalent description of the state assuming all job sizes are in
$(0, 1]$ is through a cumulative function of job sizes.
For a set of jobs (job sizes) $\mathcal{A}$, define a function 
$f_{\mathcal{A}}: [0, 1] \to \mathds{N}$ 
as
\begin{equation}
f_A(s) = |{x \in A: R_x < s}|.
\end{equation}
If we know $f_{\mathcal{A}}(s)$ for any $s \in (0, 1]$ 
then we also know $\mathcal{A}$. 
To describe $\state{S}(t)$ in our case,  we can use its equivalent 
representation using functions
$f_{\mathcal{Q}(t)}(s)$ and $f_{\mathcal{H}_\ell(t)}(s)$ for $\ell \in \mathcal{L}$.
The space of those functions is a Skorokhod space \cite{Billingsley1999}, 
for which, under the appropriate topology, we can show it is a Polish 
space \cite{ethier2009markov}. 
Our space is the product of $L+1$ of those Polish spaces and under the 
product topology, is also a Polish space.

The evolution of states over time defines a time homogeneous Markov chain,
for which we can prove its stability, by applying Theorem 1 of
\cite{Foss2004}, which we repeat next for convenience.

\begin{theorem*}\label{thm:stab_lyap}
	Let $\mathcal{X}$ be a Polish space and 	$V: \mathcal{X} \to \mathds{R}_+$ be a measurable function 
 with 
	$\sup_{x \in \mathcal{X}} V(x) = \infty$, 
	which we will refer to as Lyapunov function. Suppose there are two more 
	measurable functions $g: \mathcal{X} \to \mathds{N}$
	and $h: \mathcal{X} \to \mathds{R}$ with the following 
	properties:
	
	\begin{equation}
	\begin{aligned}
	&\inf_{x \in \mathcal{X}} h(x) > -\infty \\
	&\liminf_{V(x)\to \infty} h(x) > 0 \\
	&\sup_{V(x) \le N} g(x) < \infty, \quad \forall N > 0 \\
	&\limsup_{V(x)\to \infty} g(x)/h(x) < \infty
	\end{aligned}
	\end{equation}
	Suppose the drift of $V$ satisfies the following
	property in which $\mathds{E}_x[\cdot]$ is the conditional expectation, given  $X(t)=x$,
	\begin{equation}
	\mathds{E}_x\left[V(X(t+g(x)))-V(X(t))\right] \le -h(x).
	\end{equation}
	We define the return time to set 
	$\mathcal{X}_N = \{x \in \mathcal{X}: V(x) < N\}$ to be 
	\begin{equation}
	\tau_N = \inf\{n>0: V(X(t + n)) \le N\}
	\end{equation}
	
	Given the above, it follows that there is $N_0 > 0$, such that 
	for any $N > N_0$ and $x \in \mathcal{X}$, we have that 
	$\mathds{E}_{x} [\tau_N] < \infty$
\end{theorem*}

The theorem states that under certain conditions, the chain is 
positive recurrent to a certain subset of states for which Lyapunov 
function is bounded. From this, it can be inferred that the 
expected value of that function as time goes to $\infty$ is bounded \cite{Foss2004}.

In our proof, we pick as Lyapunov function the sum of 
sizes of jobs in the system divided by $\mu$. 
Proving that the size of jobs in system is bounded implies
that the number of jobs is also bounded under the theorem's 
assumption that job sizes have a lower bound.
If $R_i$ is the size of a job $i$, 
then the Lyapunov function is defined as
\begin{equation}
V(\state{S}(t)) \equiv V(t)= \sum_{i \in \mathcal{Q}(t)\bigcup \mathcal{H}(t)} R_i/\mu.
\end{equation}

Consider a time interval $[t_0, t_0+g(\state{S}(t_0))]$ and that
the state $\state{S}(t_0)$ is known. We want the drift to be 
negative over this time interval, for $V(t_0)$ large enough. 
Next we specify a function $g(\state{S}(t_0))$ 
and a function $h(\state{S}(t_0))$, that ensure conditions in 
Subtheorem~\ref{thm:stab_lyap} hold.
For $g(\state{S}(t_0))$ it is sufficient to assume its value is constant,
so in what follows we will have to specify this value which we will
denote by $N_2$.

We start by defining an event based on which we will differentiate 
the states with negative expected drift over $N_2$ time slots.

\begin{definition}\label{def:E_event}
	$E_{\state{S}(t_0), N_1, N_2}$ is the event
	that in time interval $[t_0, t_0 + N_2]$, every server will become
	less than half full for at most $N_1$ time units, for some 
	$N_1, N_2 \in \mathds{N}$, given the initial state $\state{S}(t_0)$.
\end{definition}

Next, we will pick the values $N_1, N_2$ such that the event
$E_{\state{S}(t_0),N_1,N_2}$ is almost certain when the 
total size of jobs in queue is large enough.

Let $t_{a, \ell}$ be the first time slot after $t_0$ that the server 
$\ell$ is less than half full and $t_{e,\ell}$ be the time after $t_0$ 
that the server empties.
Also define $Z_{\le 1/2}(t) = \sum_{j \in \mathcal{Q}(t) | R_j \le 1/2} R_j$ 
to be the sum of resource requirements of jobs in queue, whose
resource is not larger than $1/2$ and respectively 
$Z_{> 1/2}(t)$ is defined as the sum of resource requirements of the rest 
of the jobs in the queue.

The probability of $E_{\state{S}(t_0), N_1, N_2}$ can be bounded as
\begin{equation}\label{eq:help}
\begin{aligned}
&\mathds{P}\left( E_{\state{S}(t_0), N_1, N_2} \right) \ge \\
&\max\Big(\mathds{1}(Z_{\le 1/2}(t_0) > L N_2), \\
&\prod_{\ell \in \mathcal{L}} \mathds{P}\left(t_{e,\ell}-t_{a,\ell}
< N_1\right) \mathds{1}(Z_{> 1/2}(t_0) > L N_2)\Big).
\end{aligned}
\end{equation}

The logic behind the bound is that $E_{\state{S}(t_0),N_1,N_2}$ 
will be certainly true in one of the following two cases:
\begin{enumerate}
	\item If $(Z_{\le 1/2}(t_0) > LN_2)$ then the server will be more 
	than half full in next $N_2$ time slots. This is because if a server
	is less than half full, there will be at least one job whose resource
	will be less that $1/2$ that can fit in the server and those jobs are 
	enough so that there will be at least one job available in next $N_2$ time 
	slots.
	\item Once a server gets empty,	it will start serving all jobs of size 
	more than $1/2$ that are in queue at that time. While there is a job of 
	that kind that fits, the server will be more than half full.
	That will always be true in a time window of $N_2$ time slots after time
	$t_0$ if $Z_{> 1/2}(t_0) > L N_2$. That guarantees that all servers will
	have access to at least $L(N_2-T)$ jobs of size greater than $1/2$ 
	at time slot $t_0+T$. If they schedule the largest one at that time 
	slot, after getting empty, they will be able to schedule at least the 
	largest one out of the jobs that remain, in the next time slot. 
	In this case we only need to bound 
	$\mathds{P}\left(t_{e,\ell}-t_{a,\ell} < N_1\right)$ 
	which is the probability that the server will become empty in $N_1$ 
	time slots after being half empty.
\end{enumerate}
Note that if 
$\sum_{j \in \mathcal{Q}(t)} R_j > 2LN_2$,
then, by \dref{eq:help},
$\mathds{P}\left( E_{\state{S}(t_0), N_1, N_2} \right) \ge 
\prod_{\ell \in \mathcal{L}} \mathds{P}\left(t_{e,\ell}-t_{a,\ell}
< N_1\right)$. Next, we compute a lower bound on
$\mathds{P}\left(t_{e,\ell}-t_{a,\ell} < N_1\right)$.

Given that each job in service may depart during a time slot with probability
$\mu$ and there are at most $\lfloor 1/u \rfloor = K_{max}$ in a server,
we have the following bound:
\begin{equation}
\begin{aligned}
\mathds{P}\left(t_{e,\ell}-t_{a,\ell} < N_1 \right) \ge
1 - \left(1 - \mu^{K_{max}}\right)^{N_1}.
\end{aligned}
\end{equation}

If we want $\mathds{P}\left( E_{\state{S}(t_0), N_1, N_2} \right) > 1 - \epsilon_1$
it suffices to choose $N_1$ such that
\begin{equation}
\left(1 - \left(1 - \mu^{K_{max}}\right)^{N_1}\right)^L > 1 - \epsilon_1.
\end{equation}
Using the inequality $(1-x)^n > 1-nx$ for $n>0$ and $x<1$, we therefore need 
\begin{equation}\label{eq:N1}
N_1 > \frac{\log(\epsilon_1/L)}{\log(1 - \mu^{K_{max}})}.
\end{equation}

Next we give an upper bound on the maximum supportable workload, 
which we will use for comparison to the maximum workload supported
by \algA{}.
\begin{lemma*}
The maximum value of $\rho^\star$ is at most $\frac{L}{\bar R}$.
\end{lemma*}
\begin{proof}  
	Let $U(t)$ to be the total sum of the job sizes (job resource requirements) 
	in the system at time $t$, i.e.,
	\begin{equation}
	U(t) = \sum_{j \in \calQ(t) \cup \calH(t)} R_j.
	\end{equation}

	It is easy to check that $\mathbb E[U(t+1)-U(t)|\calQ(t), \calH(t)] \geq \lambda \bar{R}- \mu L$, 
	where we have used the fact that the total sum of the job sizes in all the servers in the cluster 
	is at most $L$. The system will certainly be unstable in the sense that $U(t) \to \infty$, with 
	probability one, if $\lambda \bar{R}-\mu L>0 $ or equivalently $\rho > \frac{L}{\bar R}$ (see e.g. 
	Theorem~11.3 in \cite{tweedie1976criteria}). This in turn implies that $Q(t) \to \infty$ as 
	$Q(t) \ge U(t)$ and that the system is unstable for any $\rho > \frac{L}{\bar R}$ .
\end{proof} 

Next we will show that for any 
$\epsilon > 0$, the workload $\rho$ will be supportable by our algorithm if 
$\rho < (1 - \epsilon)\frac{L}{2\bar R}$. This essentially proves that 
the best supportable workload is at least half of the optimal.

The drift of the Lyapunov function
over $N_2$ time slots, given the initial state $\state{S}(t_0)$, 
is the difference between the sizes of jobs that arrive and 
the size of jobs that depart normalized by a factor $1/\mu$.
The expected average size of arrivals in one step is 
$
\lambda \bar{R},
$
where $\lambda$ is the average number of arrivals and $\bar{R}$ 
the average job size. 
Further, the expected size of departures at time $t$ given initial state
$\state{S}(t_0)$ is
$
\mathds{E}\Big[\sum_{\ell \in \mathcal{L}}
\sum_{j \in \mathcal{H}_\ell(t)} R_j | \state{S}(t_0) \Big] \mu.
$
Hence, we can compute the drift as
\begin{equation}
\begin{aligned}
&\mathds{E}[V(t_0+N_2) - V(t_0) | \state{S}(t_0)] = \\
&N_2 \rho \bar{R}
- \mathds{E}\Big[\sum_{t=t_0}^{t_0+N_2-1}\sum_{\ell \in \mathcal{L}}
\sum_{j \in \mathcal{H}_\ell(t)} R_j | \state{S}(t_0) \Big].
\end{aligned}
\end{equation}

According to Subtheorem~\ref{thm:stab_lyap} 
we want to find $h(\state{S}(t_0))$ such that
$-h(\state{S}(t_0)) \ge \mathds{E}[V(t_0+N_2) - V(t_0) | \state{S}(t_0)]$.
Obviously we can choose this function $h(\state{S}(t_0))$ such that
$\inf_{\state{S}(t_0)} h(\state{S}(t_0)) = \inf_{\state{S}(t_0)} \mathds{E}[V(t_0+N_2) - V(t_0) | \state{S}(t_0)] \ge 
-N_2 \rho \bar{R} > -\infty$.
Now given $V(t_0) > \frac{2L N_2}{\mu}$ we have

\begin{equation}
\begin{aligned}
&- N_2 \rho \bar{R}
+ \\
&\mathds{E}\Big[\sum_{t=t_0}^{t_0+N_2-1}\sum_{\ell \in \mathcal{L}}
\sum_{j \in \mathcal{H}_\ell(t)} R_j | \state{S}(t_0), 
V(t_0) > \frac{2L N_2}{\mu}  \Big]
\ge \\
&- N_2 \rho \bar{R} + \mathds{P}\left( E_{\state{S}(t_0),N_1,N_2} 
| V(t_0) > \frac{2L N_2}{\mu} \right) \\
&\mathds{E}\left[\sum_{t=t_0}^{t_0+N_2-1}
\sum_{\ell \in \mathcal{L}} \sum_{j \in \mathcal{H}_\ell(t)} R_j
| E_{\state{S}(t_0),N_1,N_2}\right] 
\ge^{(a)} \\
&- N_2 \rho \bar{R} + (1 - \epsilon_1) 
(N_2-L N_1) \sum_{\ell \in \mathcal{L}} 1/2,
\end{aligned}
\end{equation}
where (a) is due to the fact that for a duration of at least
$(N_2-LN_1)$ time slots, all servers will be at least half full,
which is a consequence of Definition~\ref{def:E_event}.
Therefore, for $V(t_0) > \frac{2L N_2}{\mu}$, $h(\state{S}(t_0))$ is given by
\begin{equation}
h(\state{S}(t_0)) = - N_2 \rho \bar{R} + (1 - \epsilon_1) 
(N_2-L N_1) \sum_{\ell \in \mathcal{L}} 1/2.
\end{equation}

Now if we need $h(\state{S}(t_0)) > \delta$ when $V(t_0) > \frac{2L N_2}{\mu}$,
it suffices that 
\begin{equation}
\begin{aligned}
&- N_2 \rho \bar{R} + (1 - \epsilon_1) 
(N_2-L N_1) \sum_{\ell \in \mathcal{L}} 1/2 > \delta,
\end{aligned}
\end{equation}
from which it follows
\begin{equation}\label{eq:rho_c}
\begin{aligned}
&\rho  < \frac{(1-\epsilon_1)(N_2-LN_1)L/2 - \delta}{N_2 \bar{R}}.
\end{aligned}
\end{equation}

Earlier we required that $\rho < (1 - \epsilon)\frac{L}{2\bar R}$, so from
Equation~(\ref{eq:rho_c}) we get the following sufficient 
condition for the drift to be negative:
\begin{equation}\label{eq:eps_cond}
(1 - \epsilon) <
{(1 - \epsilon_1)\left(1 - L N_1/N_2 \right)}
- \frac{2\delta}{L N_2}.
\end{equation}

We can choose parameters $\epsilon_1, N_2, \delta$ so that \dref{eq:eps_cond} is true.
The choice is not unique but the following are sufficient:

\begin{equation}\label{eq:param_cond}
\epsilon_1 = \epsilon/3, \ N_2 = \lceil 3 L N_1/\epsilon \rceil, \  \delta = L N_2 \epsilon/3.
\end{equation}

This gives  the following expressions for $g(\state{S}(t_0))$ and $h(\state{S}(t_0))$.
\begin{equation}\label{eq:bf-gh}
\begin{aligned}
& h(\state{S}(t_0)) = \left\{
\begin{array}{ll}
L N_2 \epsilon/3 & V(t_0) > \frac{2L N_2}{\mu} \\
-N_2 \rho \bar{R} & \mathrm{otherwise} \\
\end{array}
\right. \\
& g(\state{S}(t_0)) = N_2 = \lceil 3 L N_1/\epsilon \rceil \\
& N_1 > \frac{\log(\epsilon/(3L))}{\log(1 - \mu^{K_{max}})}. 
\end{aligned}
\end{equation}

\subsection{All Subcases of Proof of Proposition~\ref{prop:2/3}}\label{pr:prop_2/3}
We will analyze the remaining subcases that were not analyzed in the main proof.
They all fall under the assumption that $\sum_{i \in Z_{1}} {k}^{(X)}_{i} = 1$.
We also notice that in this case $\sum_{i \in Z_{0}} {k}^{(X)}_i = 0$,
otherwise capacity constraints are not satisfied.

We further distinguish three cases for the relative size of
$Q_{1}$ compared to $U$: $Q_{1} \ge 2U/3$, $2U/3 > Q_{1} \ge U/2$
and $U/2 > Q_{1}$

\textbf{Case 2.1.} $Q_{1} \ge 2U/3$:
Consider any $\mathbf{k} \in {\mathcal{K}}_{RED}^{(J)}$ such that $k_{1} = 1$.
For that $\mathbf{k}$, it follows that
\begin{equation}
\langle \mathbf{k}, \mathbf{Q} \rangle \ge k_{1} Q_{1} \ge
2U/3 = 2/3 \langle {\mathbf{k}}^{(X)}, \mathbf{Q}^{(X)}\rangle.
\end{equation}
This means that (\ref{eq:23_opt}) is satisfied for such a choice of 
$\mathbf{k}$.

\textbf{Case 2.2.} $2U/3 > Q_{1} \ge U/2$:

For this case we need to further consider two different outcomes for
the value of $\sum_{i \in Z_{2}} k_{i}^{(X)}$ which can be either 
$0$ or $1$. The first case was analyzed in the main proof so the analysis 
for $\sum_{i \in Z_{2}} k_{i}^{(X)} = 1$ follows here.

If $Q_{2} \ge U/3$ then configuration $2 \mathbf{e}_{2}$
has weight more than $2U/3$ and is the configuration we are looking for.
If not let $U^\prime = U - Q_{1} - Q_{2}$.
Then at least one of the following has to be true
\begin{equation}\label{eq:Q_up2}
\begin{aligned}
&Q_{2m} \ge U^\prime/2^{m-2}, \quad m = 2, \cdots, J-1 \\
&Q_{2m+1} \ge U^\prime/2^{m-1}, \quad m = 1, \cdots, J-1 \\
\end{aligned}
\end{equation}
If this is not the case, then we reach a contradiction as follows
\begin{equation}\label{eq:contradiction_222}
\begin{aligned}
&U^\prime = \langle {\mathbf{k}}^{(X)}, \mathbf{Q}^{(X)}\rangle
- Q_{1} - Q_{2} = \\
&\sum_{m=2}^{J-1} \sum_{i_0 \in Z_{2m}} {k}_{i_0}^{(X)} Q_{i_0}^{(X)} +
\sum_{m=1}^{J-1} \sum_{i_1 \in Z_{2m+1}} {k}_{i_1}^{(X)} Q_{i_1}^{(X)} < \\
&\sum_{m=2}^{J-1} {k}_{i_0}^{(X)} U^\prime/2^{m-2} +
\sum_{m=1}^{J-1} {k}_{i_1}^{(X)} U^\prime/2^{m-1} \le \\
& 6 \left(\sum_{m=2}^{J-1} \sum_{i_0 \in Z_{2m}} {k}_{i_0}^{(X)} 
2/3 \times 1/2^m + \right.\\
&\left. \sum_{m=1}^{J-1} \sum_{i_1 \in Z_{2m+1}} {k}_{i_1}^{(X)}
1/2 \times 1/2^m \right) < U^\prime
\end{aligned}
\end{equation}

In the last inequality, we applied the capacity constraint that the jobs
in configuration  other than the type-$1$ and type-$2$ should fit
in a space of at most $1/6$ (as the rest is covered by the
aforementioned jobs that we know they appear once in configuration).
In other words:
\begin{equation}
\begin{aligned}
& \sum_{m=2}^{J-1} \sum_{i_0 \in Z_{2m}} {k}_{i_0}^{(X)} 
2/3 \times 1/2^m +
\sum_{m=1}^{J-1} \sum_{i_1 \in Z_{2m+1}} {k}_{i_1}^{(X)}
1/2 \times 1/2^m < \\
& \sum_{m=2}^{J-1} \sum_{i_0 \in Z_{2m}} {k}_{i_0}^{(X)} \xi_{i_0} +
\sum_{m=1}^{J-1} \sum_{i_1 \in Z_{2m+1}} {k}_{i_1}^{(X)}
\xi_{i_1} \le 1/6 
\end{aligned}
\end{equation}

The configurations that satisfy inequality~(\ref{eq:23_opt})
depend on which of the inequalities in \dref{eq:Q_up2} is true.

If $Q_{2m} \ge U^\prime/2^{m-2}$ for some $m \in [2, \cdots, J-1]$ then 
inequality \dref{eq:23_opt} is true either for configuration $2 \mathbf{e}_{2}$
if $Q_2 \ge U/3$ or for configuration 
$\mathbf{e}_{1} + \lfloor 2^m/3 \rfloor \mathbf{e}_{2m}$ otherwise as
\begin{equation}
\begin{aligned}
\langle \mathbf{k}, \mathbf{Q} \rangle &=
Q_{1} + \lfloor 2^m/3 \rfloor Q_{2m} \ge
Q_{1} + 2^{m-2} Q_{2m} \\
&\ge Q_{1} + U^\prime = U - Q_{2} > 2U/3.
\end{aligned}
\end{equation}

Similarly if $Q_{2m+1} \ge U^\prime/(3 \cdot 2^{m-1})$
for some $m \in [1, \cdots J-1]$ then inequality \dref{eq:23_opt} is true
either for configuration $2 \mathbf{e}_{2}$
if $Q_2 \ge U/3$ or for configuration 
$\mathbf{e}_{1} + 2^{m-1} \mathbf{e}_{2m+1}$ as

\begin{equation}
\begin{aligned}
\langle \mathbf{k}, \mathbf{Q} \rangle &= 
Q_{1} + 2^{m-1} Q_{2m+1} \\
& \ge Q_{1} + U^\prime U \ge U - Q_{2} > 2U/3.
\end{aligned}
\end{equation}

\begin{align}
&Q_{1} + \lfloor 2^m/3 \rfloor Q_{2m}
\ge Q_{1} + 2^{m-2} Q_{2m} \ge \\
&Q_{1} + U^\prime \ge U - Q_{2} > 2U/3.
\end{align}

\textbf{Case 2.3.} $Q_{1} < U/2$:

At least one of the following inequalities is true:
\begin{equation}
\begin{aligned}
&Q_{2m} \ge 2U/3 \times 1/2^{m}, \quad m=1,\cdots, J-1 \\
&Q_{2m+1} \ge U/2 \times 1/2^{m}, \quad m=1,\cdots, J-1
\end{aligned}
\end{equation}

The conditions are the same as those of (\ref{eq:Q_U}) except that
$Q_{0}$ is not included now. We can again use proof by contradiction
as in (\ref{eq:contradiction1}) and get
\begin{equation}\label{eq:contradiction23}
\begin{aligned}
&U = \langle {\mathbf{k}}^{(X)}, \mathbf{Q}^{(X)} \rangle = \\
&\sum_{m=1}^{J-1} \sum_{i_0 \in Z_{2m}} {k}_{i_0}^{(X)} Q_{i_0}^{(X)} +
\sum_{m=0}^{J-1} \sum_{i_1 \in Z_{2m+1}} {k}_{i_1}^{(X)} Q_{i_1}^{(X)} < \\
&\left(\sum_{m=1}^{J-1} \sum_{i_0 \in Z_{2m}} {k}_{i_0}^{(X)} 2/3
\times 1/2^{m} + \right. \\
&\left. \sum_{m=0}^{J-1} \sum_{i_1 \in Z_{2m+1}} {k}_{i_1}^{(X)} 1/2 \times 1/2^{m} 
\right) U < U.
\end{aligned}
\end{equation}
Again the last inequality is due to the capacity
constraint of the server under the assumption that
$\sum_{i \in Z_{0}} {k}_{i}^{(X)} = 0$ and 
$\sum_{i \in Z_{1}} {k}_{i}^{(X)} = 1$.

Now if $Q_{2m} \ge 2U/3 \times 1/2^{m}$ for some $m=1,\cdots, J-1$
then configuration $2^m \mathbf{e}_{2m}$ will satisfy \dref{eq:23_opt}
while if $Q_{2m+1} \ge U/2 \times 1/2^{m}$ for some $m=1,\cdots, J-1$
then configuration $3 \cdot 2^{m-1} \mathbf{e}_{2m}$ 
will satisfy \dref{eq:23_opt}.

\subsection{Proof of Theorem~\ref{thm:VQS1}}\label{pr:thm_VQS1}

We need to show that for any $\rho$ such that 
$\rho < 2/3 \rho^\star$ the system is stable. 
For this we will first prove the following Lemma 
which is a consequence of Proposition~\ref{prop:2/3}.
\begin{lemma*}\label{prop:2/3rho}
	If $\rho < 2/3 \rho^\star$, then there is 
	$\mathbf{x} \in \mathrm{Conv}({\mathcal{K}}_{RED}^{(J)})$ and $\epsilon>0$, 
	such that $\bm{\rho} < (1-\epsilon) L \mathbf{x}$, where 
	$\bm{\rho} = \rho \bm{P}^{(I)}$, $I$ is the partition (\ref{eq:VQAintervals})
	and $L$ is the number of servers.
\end{lemma*}
\begin{proof}
	In the proof of Theorem~\ref{thm:OPT} we constructed a sequence of partitions
	$X^{(n)}$, $n \in \mathds{N}$, whose maximum supportable workload approaches $\rho^\star$. In this proof we will consider the partitions $X^{+(n)}$ 
	generated from partitions $X^{(n)}$ the way described next.
	To simplify description we provide the definition of $X^{+(n)}$ only for
	the case of continuous probability distribution function although definition
	and result can be generalized. 
	As a reminder, for the continuous probability distribution case 
	of Theorem~\ref{thm:OPT}, $X^{(n)}$ is a collection of intervals
	$X^{(n)}_i: (\xi_{i-1}^{(n)}, \xi_{i}^{(n)}]$ for $i=1, \ldots, 2^{n+1}$. 
	Partition $X^{+(n)}$, $n \in \mathds{N}$ will be a collection of $c_n$
	intervals $X_i^{+(n)}:(\xi_{i-1}^{+(n)}, \xi_{i}^{+(n)}]$ where
	$\xi_{0}^{+(n)} = 0$, $\xi_{c_n}^{+(n)} = 1$ and $\xi_{i}^{+(n)}$
	is the $i$-th largest element in the set 
	$\{\xi_{i^\prime}^{(n)}, i^\prime=1, \ldots, 2^{n+1}-1\} \cup 
	\{1/2^m, m=0,\cdots,J-1 \} \cup \{2/3 \times 1/2^m, m=0, \cdots, J-1 \}$
	for $i=1,\ldots, c_n-1$.

	The partition $X^{+(n)}$ is finer than $X^{(n)}$ so
	$\overline{\rho}^\star({X^{(n)}}) \le \overline{\rho}^\star({X^{+(n)}})$ 
	and 
	\begin{equation}\label{eq:limit}
	\lim_{n \to \infty} \overline{\rho}^\star({X^{+(n)}}) = \rho^\star.
	\end{equation}
	
	Consider now any $\rho < \overline{\rho}^\star({X^{+(n)}})$ so 
	the workload vector $\rho\bm{P}^{(X^{+(n)})}$ is supportable under 
	the assumption of upper-rounded virtual queues.

	Next we define sets similar to sets $Z_j$ of Proposition~\ref{prop:2/3}
	which we denote by $Z_j^{+(n)}$ for $j=1, \ldots, 2J$, $n \in \mathds{N}$.
	We will have that for $i=1, \ldots, c_n$, 
	$i \in Z_j^{+(n)}$ iff $\xi^{+(n)}_i \in I_j$ where
	$I_j$ are the intervals defined in (\ref{eq:VQAintervals}).
	We also define 
	$\underline{i}_j = \arg\min_{i \in Z_j^{+(n)}} \xi^{(n)}_i$
	and the probability vector $\bm{\underline P}^{(X^{+(n)})}$ as
	\begin{equation}\label{eq:P}
	\underline{P}^{(X^{+(n)})}_{i} = \left\{ 
	\begin{array}{ll}
	\sum_{i^\prime \in Z_j^{+(n)}} {P}^{(X^{+(n)})}_{i^\prime} & 
	\text{if} \quad i = \underline{i}_j \\
	0 & \text{otherwise}
	\end{array}
	\right.
	\end{equation}
	
	We see that the workload vector $\rho \bm{\underline P}^{(X^{+(n)})}$ 
	is supportable if $\rho \bm{P}^{(X^{+(n)})}$ is supportable. 
	This is because the original workload $\rho \bm{P}^{(X^{+(n)})}$
	is equivalent to $\rho \bm{\underline P}^{(X^{+(n)})}$ 
	if job sizes of all arriving job are modified the following way.
	Jobs that join the $\rm{VQ}_i$, where $i \in Z_j^{+(n)}$ and
	$i \neq \underline{i}_j$, reduce their size such that they join 
	$\rm{VQ}_{\underline{i}_j}$ instead, while if $i = \underline{i}_j$ 
	job sizes remain unchanged. Since all job sizes are reduced or remain 
	the same and system is stable without this change, then system should
	also be stable with this modification.
	
	Similarly to (\ref{eq:P}), let $\bm{\underline Q}^{(X^{+(i)})}$ be defined as
	\begin{equation}
	{\underline{Q}}^{(X^{+(n)})}_{i} = \left\{ 
	\begin{array}{ll}
	\sum_{i^\prime \in Z_j^{+(n)}} {Q}^{(X^{+(n)})}_{i^\prime} & 
	\text{if} \quad i = \underline{i}_j \\
	0 & \text{otherwise}
	\end{array}
	\right.
	\end{equation}
	
	If $\overline{\mathcal{K}}^{+(n)}$ is the set of feasible
	configurations under upper-rounded virtual queues assumption
	for partition $X^{+(n)}$, we should also have 
	\begin{equation}\label{eq:conv}
	\overline{\rho}^\star({X^{+(n)}}) \bm{P}^{(X^{+(n)})} \le L \bm{x}, \ 
	\bm{x} \in \mathrm{Conv}(\overline{\mathcal{K}}^{+(n)}).
	\end{equation}

	If now $\rho < 2/3 \rho^\star$, it means 
	\begin{equation}\label{eq:rho}
	\exists i \in \mathds{N}: \rho < 2/3\overline{\rho}^\star({X^{+(i)}}),
	\end{equation}
	because of (\ref{eq:limit}). 
	Eventually we have
	\begin{equation}\label{eq:rhoP1}
	\begin{aligned}
	&\langle \rho \bm{P}^{(I)}, \mathbf{Q}^{(I)} \rangle =^{(a)}
	\langle \rho \bm{\underline P}^{(X^{+(n)})}, 
	\mathbf{\underline Q}^{(X^{+(n)})} \rangle <^{(b)} \\
	&2/3\langle \overline{\rho}^\star({X^{+(n)}}) \bm{\underline P}^{(X^{+(n)})}, 
	\mathbf{\underline Q}^{(X^{+(n)})} \rangle \le^{(c)}  \\
	&2/3 \langle L \bm{x}, \mathbf{\underline Q}^{(X^{+(n)})} \rangle.
	\end{aligned}
	\end{equation}
	Equality (a) follows from the fact that vectors $\bm{P}^{(I)}$
	and $\bm{\underline P}^{(X^{+(n)})}$ are identical, if $0$ entries
	are ignored in the latter vector, while the same property is true 
	for vectors $\mathbf{Q}^{(I)}$ and $\mathbf{\underline Q}^{(X^{+(n)})}$. 
	Then (b) follows from (\ref{eq:rho}) 
	and in (c) we used $\bm{x}$ from (\ref{eq:conv}). 
	Given that $\mathrm{Conv}(\overline{\mathcal{K}}^{+(n)})$ is a convex set, 
	there should be a $\mathbf{k}^{(X^{+(n)})} \in \overline{\mathcal{K}}^{+(n)}$ such that
	\begin{equation}\label{eq:rhoP2}
	2/3\langle L \bm{x}, \mathbf{\underline Q}^{(X^{+(n)})} \rangle \le
	2/3\langle L \mathbf{k}^{(X^{+(n)})}, \mathbf{\underline Q}^{(X^{+(n)})} \rangle,
	\end{equation}
	and eventually because of Proposition~\ref{prop:2/3}, there is a 
	${\mathbf{k}} \in {\mathcal{K}}_{RED}^{(J)}$ such that
	\begin{equation}\label{eq:rhoP3}
		2/3\langle L \mathbf{k}^{(X^{+(n)})}, \mathbf{\underline Q}^{(X^{+(n)})} \rangle
		\le L\langle \mathbf{k}, \mathbf{Q}^{(I)} \rangle.
	\end{equation}
	
	Using (\ref{eq:rhoP1}), (\ref{eq:rhoP2}) and (\ref{eq:rhoP3}), it follows  
	that for any virtual queue size vector $\mathbf{Q}^{(I)}$ under partition 
	$I$, there exists $\mathbf{k} \in {\mathcal{K}}_{RED}^{(J)}$ such that
	\begin{equation}
		\langle \rho \bm{P}^{(I)}, \mathbf{Q}^{(I)} \rangle <
		L\langle \mathbf{k}, \mathbf{Q}^{(I)} \rangle.
	\end{equation} 
	As a result, $\rho \bm{P}^{(I)}$ is in the interior of 
	$\mathrm{Conv}({\mathcal{K}}_{RED}^{(J)})$ so there is 
	$\mathbf{x} \in \mathrm{Conv}({\mathcal{K}}_{RED}^{(J)})$ and $\epsilon > 0$
	such that
	\begin{equation}
		\bm{\rho} = \rho \bm{P}^{(I)} < (1 - \epsilon) L \mathbf{x}.
	\end{equation}
\end{proof}

The \textit{state} of our system at time slot $t$ is
\begin{equation}\label{eq:stateVQA}
\state{S}(t) = (\bm{\mathcal{Q}}(t), \bm{\mathcal{H}}(t)).
\end{equation}
$\bm{\mathcal{Q}}(t)$ is a vector of sets of jobs in queues
which is equal to $(\mathcal{Q}_{j}(t), j=0, \cdots, 2J-1)$.
The cardinality of each set of jobs is equal to the corresponding queue
size, i.e. $|\mathcal{Q}_{j}(t)| = Q_{j}(t)$.
$\bm{\mathcal{H}}(t)=(\mathcal{H}_\ell(t) \quad \ell \in \mathcal{L})$ is
the vector of sets of scheduled jobs in servers.

We will structure our proof again around the result of 
Subtheorem~\ref{thm:stab_lyap}. The Lyapunov function that we use is
\begin{equation}\label{eq:lyap-VQS}
V(\state{S}(t))\equiv V(t) = \sum_{j=0}^{2J-1} \frac{Q_{j}(t)^2}{2\mu}.
\end{equation}
Given state $\state{S}(t_0)$ at a reference time $t_0$
is known, we want to describe functions 
$g(\state{S}(t_0))$ and $h(\state{S}(t_0))$ that satisfy the conditions of 
Subtheorem~\ref{thm:stab_lyap}.
Function $g(\state{S}(t_0))$ will be fixed and equal to $N_2$. 
The value of $N_2$ as well as the function $h(\state{S}(t_0))$ 
will be specified later.

Given that $h(\state{S}(t_0))$ has to be eventually negative we will 
differentiate the
initial states for which this will happen based on the following event.
Given the state of the system $\state{S}(t_0)$,
we define the event $E_{\state{S}(t_0),N_1,N_2,\gamma}$ to be the event
that in time interval $[t_0, t_0 + N_2]$, every server will be
scheduling according to a configuration whose weight is at most $\gamma$ fraction
of that of maximum weight configuration in $\mathcal{K}^{(J)}_{RED}$,
for at most $N_1$ time units, for some $N_1, N_2 \in \mathds{N}$.
In all the following results we assume $\gamma \in (0, 1)$ so it 
can be arbitrarily close to $1$, but strictly less than that.
The next lemma states the conditions under which event 
$E_{\state{S}(t_0),N_1,N_2,\gamma}$ is almost certain. 

\begin{lemma*}\label{lem:E}
	We can ensure $\mathds{P}(E_{\state{S}(t_0),N_1,N_2,\gamma}) > 1 - \epsilon$,
	if $N_1 >\frac{\log{(\epsilon/(2L))}}{\log{\left(1 - \mu^{K_{max}}\right)}}$
	and $\left\Vert \mathbf{Q}(t_0) \right\Vert > B_{\gamma} \frac{N_2}{\epsilon}$
	where $B_{\gamma}$ some constant and $\left\Vert \cdot \right\Vert$ the 
	Euclidean norm.
\end{lemma*}
\begin{subproof}
	Let $t_{e_\ell(i)}$ denote the $i$th time that server $\ell$ gets empty
	between time slots $t_0$ and $t_0 + N_2$.
	We notice that we can bound event $E_{\state{S}(t_0), N_1, N_2,\gamma}$, by
	the event that $t_{e_\ell(1)} - t_0 < N_1$ and for the next 
	$N_2 - t_{e_\ell(1)}$ time slots, the weight of configuration will always 
	be greater that $\gamma$ fraction of the optimal.
	\begin{equation}
	\begin{aligned}
	&\mathds{P}\left( E_{\state{S}(t_0), N_1, N_2,\gamma} \right) \ge
	\prod_{\ell \in \mathcal{L}} \mathds{P}\left(t_{e_\ell(1)}-t_0
	< N_1\right) \\
	&\mathds{P}\left(\langle\mathbf{k}(t_{e_\ell(i_n)}), Q(t) \rangle -
	\langle \gamma\mathbf{k}^\star(t), Q(t) \rangle > 0, \quad 
	t > t_{e_\ell(1)}\right)
	\end{aligned}
	\end{equation}
	where $t \in [t_0, t_0 + N_2]$, $i_n$ the latest time slot less than $t$
	that the server got empty and $\mathbf{k}^\star(t)$ is the max-weight
	configuration at time slot $t$.
	
	A bound for the first term is
	\begin{equation}\label{eq:P_dn}
	\begin{aligned}
	\mathds{P}\left(t_{e_\ell(1)}-t_0 < N_1 \right) \ge
	1 - \left(1 - \mu^{K_{max}}\right)^{N_1}.
	\end{aligned}
	\end{equation}
	When a server becomes empty for the $i$th time,
	the following inequality holds 
	\begin{equation}\label{eq:emptyA1}
	\langle\mathbf{k}(t_{e_\ell(i)}), \mathbf{Q}(t_{e_\ell(i)})\rangle \ge
	\langle\mathbf{k}, \mathbf{Q}(t_{e_\ell(i)}) \rangle 
	\quad \forall \mathbf{k} \in {\mathcal{K}}_{RED}^{(J)}.
	\end{equation}
	The condition $\langle \mathbf{k}(t_{e_\ell(i_n)}), Q(t) \rangle -
	\langle \gamma\mathbf{k}^\star(t), Q(t) \rangle > 0$
	will be violated if	for at least one $\mathbf{k} \in {\mathcal{K}}_{RED}^{(J)}$,
	\begin{equation}\label{eq:emptyA2}
	\langle\mathbf{k}(t_{e_\ell(i)}), \mathbf{Q}(t) \rangle <
	\gamma \langle\mathbf{k}, \mathbf{Q}(t) \rangle.
	\end{equation}
	
	As a result, for this particular $\mathbf{k} \in {\mathcal{K}}_{RED}^{(J)}$,
	we have,
	\begin{equation}\label{eq:DQ}
	\begin{aligned}
	&\left \Vert \mathbf{Q}(t_{e_\ell(i)}) - \mathbf{Q}(t)\right \Vert > \\
	&\left \Vert \mathbf{Q}(t_{e_\ell(i)}) \right \Vert
	\frac{\left | \langle \mathbf{k}(t_{e_\ell(i)}) - \mathbf{k},
		\mathbf{k}(t_{e_\ell(i)}) - \gamma \mathbf{k} \rangle \right |}
	{\left \Vert \mathbf{k}(t_{e_\ell(i)}) - \mathbf{k} \right \Vert
		\left \Vert \mathbf{k}(t_{e_\ell(i)}) - \gamma \mathbf{k} \right \Vert}.
	\end{aligned}
	\end{equation}
	Given that $\mathbf{A}[t_1, t_2]$ is the vector of
	arrivals per type in time interval $[t_1, t_2)$, $A[t_1, t_2]$
	the absolute number of arrivals in the same interval and
	$\mathbf{D}[t_1, t_2]$, ${D}[t_1, t_2]$ the respective
	values for departures, then
	\begin{equation}\label{eq:AD}
	\begin{aligned}
	&\left \Vert \mathbf{Q}(t_{e_\ell(i)}) - \mathbf{Q}(t) \right \Vert =
	\left \Vert \mathbf{A}[t_{e_\ell(i)}, t] - \mathbf{D}[t_{e_\ell(i)}, t]
	\right \Vert \le \\
	&\left \Vert \mathbf{A}[t_{e_\ell(i)}, t] \right \Vert +
	\left \Vert \mathbf{D}[t_{e_\ell(i)}, t] \right \Vert \le \\
	&{A}[t_{e_\ell(i)}, t] + {D}[t_{e_\ell(i)}, t]
	\end{aligned}
	\end{equation}
	
	Setting
	\begin{equation}\label{eq:Cb}
	B_{\gamma1} = \min_{\mathbf{k^\prime}, \mathbf{k}}
	\frac{\left | \langle \mathbf{k^\prime} - \mathbf{k},
		\mathbf{k^\prime} - \gamma \mathbf{k} \rangle \right |}
	{\left \Vert \mathbf{k^\prime} - \mathbf{k} \right \Vert
		\left \Vert \mathbf{k^\prime} - \gamma \mathbf{k} \right \Vert}
	\end{equation}
	we can as a result claim that
	\begin{equation}\label{eq:P_AD}
	\begin{aligned}
	&\mathds{P}\left(
	\langle\mathbf{k}(t_{e_\ell(i_n)}), Q(t)\rangle -
	\langle \gamma\mathbf{k}^\star(t), Q(t) \rangle > 0, t > t_{e_\ell(1)}\right) > \\
	&\mathds{P}\left( {A}[t_{e_\ell(1)}, t] + {D}[t_{e_\ell(1)}, t] <
	B_{\gamma1} \left \Vert \mathbf{Q} (t_{e_\ell(1)}) \right \Vert \right) > \\
	&\mathds{P}\left( {A}[t_0, t_0 + N_2] + {D}[t_0, t_0 + N_2] <
	B_{\gamma1} \left \Vert \mathbf{Q} (t_0) \right \Vert \right) >  \\
	&1 - \mathds{P}\left( {A}[t_0, t_0 + N_2] >
	B_{\gamma1}/2 \left \Vert \mathbf{Q}(t_0) \right \Vert \right) \\
	&\mathds{P}\left( {D}[t_0, t_0 + N_2] >
	B_{\gamma1}/2 \left \Vert \mathbf{Q}(t_0) \right \Vert \right) \ge \\
	&1 - \frac{(\lambda + \mu K_{max} L)N_2}
	{B_{\gamma1}/2 \left \Vert \mathbf{Q}(t_0)\right \Vert}.
	\end{aligned}
	\end{equation}
	
	Combining Equations (\ref{eq:P_dn}) and (\ref{eq:P_AD}), if we need
	$\mathds{P}\left( E_{\state{S}(t_0),N_1,N_2,\gamma} \right) > 1 - \epsilon$,
	it suffices to have
	\begin{equation}
	\left(1 - \left(1 - \mu^{K_{max}}\right)^{N_1}\right)
	\left(
	1 - \frac{(\lambda + \mu K_{max} L)N_2}
	{B_{\gamma1}/2 \left \Vert \mathbf{Q}(t_0)\right \Vert}
	\right) > 1 - \epsilon/L
	\end{equation}
	or
	\begin{equation}
	1 - \left(1 - \mu^{K_{max}}\right)^{N_1} > 1 - \epsilon/(2L)
	\end{equation}
	which  implies
	\begin{equation}
	N_1 > \frac{\log{(\epsilon/(2L))}}
	{\log{\left(1 - \mu^{K_{max}}\right)}}
	\end{equation}
	and
	\begin{equation}
	\begin{aligned}
	& 1 - \frac{(\lambda + \mu K_{max} L)N_2}{B_{\gamma1}/2 \left \Vert \mathbf{Q}(t_0)
		\right \Vert} > 1 - \epsilon/(2L)
	\end{aligned}
	\end{equation}
	which implies
	\begin{equation}\label{eq:Qt_bound1}
	\begin{aligned}
	&\left \Vert \mathbf{Q}(t_0) \right \Vert >
	\frac{(\lambda + \mu K_{max} L)N_2}{B_{\gamma1} \epsilon /(4L)}
	\end{aligned}
	\end{equation}
	
	The Lemma is true for $B_\gamma = \frac{4L(\lambda + \mu K_{max} L)}{B_{\gamma1}}$
\end{subproof}


The change of $V(t)$ in one time slot, will be
\begin{equation}
\begin{aligned}
&V(t+1) - V(t) = \sum_{j=0}^{2J-1} 
\frac{Q_{j}(t)(A_{j}[t+1, t] - D_{j}[t+1, t])}{\mu} \\
& + \sum_{j=0}^{2J-1} \frac{(A_{j}[t+1, t] - D_{j}[t+1, t])^2}{2\mu}
\end{aligned}
\end{equation}

In one time slot we also have
$\mathds{E}(A_{j}[t+1, t]) = \lambda_{j} = \lambda p_{j}$
and
$\mathds{E}(D_{j}[t+1, t]) = \mu \sum_{\ell \in \mathcal{L}} k^{(\ell)}_{j}(t)$
where $k^{(\ell)}_{j}(t)$ is the number of jobs in server $\ell$ at time slot
$t$, that come from $\mVQ_j$. 
Setting $\lambda_{j}/\mu = \rho_{j}$, the following holds when it comes
to computing the expected change in the Lyapunov function,
\begin{equation}\label{eq:deltaV}
\mathds{E}[V(t+1) - V(t) | \state{S}(t)] \le \sum_{j=0}^{2J-1}
Q_{j}(t)\left(\rho_{j} - \sum_{\ell \in \mathcal{L}} k^{(\ell)}_{j}
\right) + B_\beta.
\end{equation}
where
\begin{equation}\label{eq:B_b}
\begin{aligned}
&B_\beta = \mu \sum_{j=0}^{2J-1} 
\left( \rho_{j}^2 + \textrm{Var}(A_{j}[t+1, t]) \right) + 2\mu J L K_{max}^2
\end{aligned}
\end{equation}
and $\textrm{Var}(\cdot)$ is the variance.
As a first step, it is obvious from (\ref{eq:deltaV}) that the
expected change is bounded, when all queue sizes at time slot $t$ are bounded.
Specifically if we ignore the effect of departures we can use the following 
bound when drift is considered over a number of $N_2$ times slots:
\begin{equation}\label{eq:bound_dV}
\mathds{E}[V(t+N_2) - V(t) | \state{S}(t)] \le N_2(\langle \mathbf{Q}(t) + N_2 \bm{\lambda}, \bm{\rho}\rangle + B_\beta).
\end{equation}

When $\left \Vert \mathbf{Q}(t) \right \Vert$ is large enough we can use the 
following lemma to derive a stricter bound.
\begin{lemma*}\label{lem:deltaV}
	If at a given time slot $t$, the weight of all configurations 
	$\mathbf{k}^\ell$, $\ell \in \mathcal{L}$ is at least $\gamma$ times
	the weight of all configurations of $\mathcal{K}^{(J)}_{RED}$ and 
	{workload} $\rho$ satisfies $\rho < 2/3 \rho^\star$ 
	then the following is true:
	\begin{equation}
	\sum_{j=0}^{2J-1}
	Q_{j}(t)\left(\rho_{j} - \sum_{\ell \in \mathcal{L}} k^{(\ell)}_{j}
	\right) < - B_\alpha \left \Vert \mathbf{Q}(t) \right \Vert 
	\end{equation}
	for some constant $B_\alpha > 0$. 
\end{lemma*}

\begin{subproof}
	
	If $\rho < 2/3 \rho^\star$ then there is $\gamma < 1$ such that 
	$\rho < 2/3 \gamma \rho^\star$.
	Then because of Lemma~\ref{prop:2/3rho}, 
	there is an  $\mathbf{x} \in Conv({\mathcal{K}}_{RED}^{(J)})$, such that
	$\bm{\rho} < (1-\epsilon)\gamma L \mathbf{x}$, where factor $L$ is due to $L$ identical servers.
	
	Let $\tilde{\mathbf{k}}^\ell(t)$ be the active configuration of server $\ell$
	at time slot $t$. Under the claim that 
	$k^{(\ell)}_{j}(t) \ge \tilde{k}^{(\ell)}_{j}(t)$ when $Q_{j}(t) > 0$,
	which is true because of the way algorithm works, we should also have
	\begin{equation}
	\begin{aligned}
	&\langle \bm{\rho}, \mathbf{Q}(t)\rangle \le (1-\epsilon)\gamma L
	\langle \mathbf{x}, \mathbf{Q}(t)\rangle \le
	(1-\epsilon) \gamma L \langle \mathbf{k}^\star(t), \mathbf{Q}(t)\rangle \\
	&\le (1-\epsilon) \sum_{\ell \in \mathcal{L}}
	\langle \mathbf{k}^{(\ell)}(t), \mathbf{Q}(t)\rangle
	\end{aligned}
	\end{equation}
	
	Using this result
	\begin{equation}
	\begin{aligned}
	&\sum_{j=0}^{2J-1} Q_{j}(t)\left(\rho_{j} -
	\sum_{\ell \in \mathcal{L}} k^{(\ell)}_{j} \right) =
	\left\langle \bm{\rho} - \sum_{\ell \in \mathcal{L}} \mathbf{k}^{(\ell)}(t),
	\mathbf{Q}(t)\right\rangle \\
	& < -\epsilon \left \langle \sum_{\ell \in \mathcal{L}}
	\mathbf{k}^{(\ell)}(t), \mathbf{Q}(t) \right \rangle
	\le - B_\alpha \left \Vert \mathbf{Q}(t) \right \Vert
	\end{aligned}
	\end{equation}
	where
	\begin{equation}\label{eq:B2}
	\begin{aligned}
	& B_\alpha = 
	\epsilon\left \Vert \sum_{\ell \in \mathcal{L}} \mathbf{k}^{(\ell)}(t)\right \Vert
	\cos\left(\sum_{\ell \in \mathcal{L}} \mathbf{k}^{(\ell)}(t), \mathbf{Q}(t)
	\right) \ge \\
	& \epsilon \gamma L \left \Vert \mathbf{k}^\star(t) \right \Vert
	\cos\left( \mathbf{k}^\star(t), \mathbf{Q}(t) \right)	\ge
	\epsilon \frac{\gamma L}{\sqrt{2J}} > 0.
	\end{aligned}
	\end{equation}
	In the above derivation we show that $B_\alpha$ is strictly positive and that can be 
	done using that the expression
	$\left \Vert \mathbf{k}^\star(t) \right \Vert
	\cos\left( \mathbf{k}^\star(t), \mathbf{Q}(t) \right)$
	is at least $\frac{1}{\sqrt{2J}}$. That can be shown by noticing
	that maximum of $\left \Vert \mathbf{k} \right \Vert
	\cos\left( \mathbf{k}, \mathbf{Q}(t) \right)$ over all
	configurations $\mathbf{k} \in {\mathcal{K}}_{RED}^{(J)}$ is 
	at least the maximum of $\cos\left( \mathbf{k}, \mathbf{Q}(t) \right)$
	over the same set of configurations, since 
	$\left \Vert \mathbf{k} \right \Vert \ge 1$ for all 
	$\mathbf{k} \in {\mathcal{K}}_{RED}^{(J)}$.
	
	Lastly one can think of 
	the expression $\cos\left( \mathbf{k}, \mathbf{Q}(t) \right)$ 
	as the projection of a unit queue vector
	onto a configuration $\mathbf{k} \in {\mathcal{K}}_{RED}^{(J)}$ 
	Given that the set of configurations ${\mathcal{K}}_{RED}^{(J)}$ spans a
	space of $2J$ dimensions, the largest cosine will have a value of at least 
	$\frac{1}{\sqrt{2J}}$.
\end{subproof}

In the last part of our proof, we will give the conditions under which
the drift over $N_2$ time slots is negative or equivalently 
$h(\state{S}(t_0))$ can be chosen to be positive.

\begin{lemma*}\label{lem:DV}
	We will have that $\mathds{E}[V(t_0 + N_2) - V(t_0) | 
	\state{S}(t_0)] < -\delta < 0$, if 
	\begin{equation}\label{eq:N2_bound}
	\begin{aligned}
	& N_2 >  \frac{L N_1(1-\epsilon)(B_\alpha + \left \Vert \bm{\rho} \right \Vert)}
	{B_\alpha - \epsilon (B_\alpha + \left \Vert \bm{\rho} \right \Vert)} \\
	\end{aligned}
	\end{equation}
	\begin{equation}\label{eq:Qt_bound2}
	\begin{aligned}
	& \left\Vert \mathbf{Q}(t_0)\right\Vert > \\
	& \frac{-\delta - N_2 B_\beta}
	{N_2 \left( -B_\alpha + \epsilon(B_\alpha + \left \Vert \bm{\rho} \right \Vert) \right)
		+ L N_1(1 - \epsilon)(B_\alpha + \left \Vert \bm{\rho} \right \Vert)}.
	\end{aligned}
	\end{equation}
\end{lemma*}

\begin{subproof}
	First we provide a bound for 
	$\mathds{E}[V(t_0 + N_2) - V(t_0) | \state{S}(t_0)]$,
	based on Lemmas \ref{lem:E} and \ref{lem:deltaV},
	\begin{equation}
	\begin{aligned}
	&\mathds{E}[V(t_0 + N_2) - V(t_0) | \state{S}(t_0)] < \\
	&\mathds{P}(E_{{S}(t_0),N_1,N_2,\gamma})
	\mathds{E}[V(t_0 + N_2) - V(t_0) | \state{S}(t_0), 
	E_{\state{S}(t_0),N_1,N_2,\gamma}] \\
	& + (1 - \mathds{P}(E_{\state{S}(t_0),N_1,N_2,\gamma}))
	N_2 \sum_{j=0}^{2J-1} Q_{j}(t_0)\rho_{j} + N_2 B_\beta < \\
	&(1-\epsilon) (N_2 - L N_1)(- B_\alpha \left \Vert \mathbf{Q}(t_0) \right \Vert)
	+ (1-\epsilon) L N_1 \left \Vert \bm{\rho} \right \Vert
	\left \Vert \mathbf{Q}(t_0) \right \Vert \\
	& + \epsilon N_2 \left \Vert \bm{\rho} \right \Vert
	\left \Vert \mathbf{Q}(t_0) \right \Vert + N_2 B_\beta = \\
	& \left( N_2(-B_\alpha + \epsilon (B_\alpha + \left \Vert \bm{\rho} \right \Vert))
	+L N_1(1 - \epsilon)(B_\alpha + \left \Vert \bm{\rho} \right \Vert) \right)
	\left \Vert \mathbf{Q}(t_0) \right \Vert\\ 
	&+ N_2 B_\beta
	\end{aligned}
	\end{equation}
	
	To ensure $\mathds{E}[V(t_0 + N_2) - V(t_0) | \state{S}(t_0)] < -\delta < 0$, 
	it suffices that (\ref{eq:Qt_bound2}) holds.
	The inequality is true provided the denominator is negative, so
	a sufficient choice of parameters 
	is $\epsilon < \frac{B_\alpha}{B_\alpha + \left \Vert \bm{\rho} \right \Vert}$ 
	and $N_2$ given by (\ref{eq:N2_bound}).
\end{subproof}

We have eventually proven that conditions of Theorem~\ref{thm:stab_lyap}
are true for the state space of our problem, when Lyapunov function is the one
in Equation~(\ref{eq:lyap-VQS}) 
and $h(\state{S}(t_0))$, $g(\state{S}t_0)$ are given by

\begin{equation}
\begin{aligned}
& h(\state{S}(t_0)) = \left\{
\begin{array}{ll}
\delta & \left \Vert \mathbf{Q}(t_0) \right \Vert > Q_{t} \\
-N_2((Q_{t} + N_2 \left\Vert\bm{\lambda}\right\Vert) 
\left\Vert\bm{\rho}\right\Vert + B_\beta) & \left \Vert \mathbf{Q}(t_0) \right \Vert \le Q_{t}
\end{array}
\right. , \\
& g(\state{S}(t_0)) = N_2 = \left\lceil 
\frac{L N_1(1-\epsilon)(B_\alpha + \left \Vert \bm{\rho} \right \Vert)}
{B_\alpha - \epsilon (B_\alpha + \left \Vert \bm{\rho} \right \Vert)} \right \rceil, \\
& N_1 > \frac{\log{(\epsilon/2)}}{\log{\left(1 - \mu^{K_{max}}\right)}}, \\
& \epsilon < \frac{B_\alpha}{B_\alpha + \left \Vert \bm{\rho} \right \Vert}, \\
\end{aligned}
\end{equation}
where $B_\alpha$ is defined in (\ref{eq:B2}), $B_\beta$ in (\ref{eq:B_b}) and
$Q_t$ is the maximum of expressions (\ref{eq:Qt_bound1}) and (\ref{eq:Qt_bound2}).

\subsection{Proof of Proposition~\ref{prop:2/3opt}}\label{pr:2/3opt}
Since the partitions are countable, we can choose an $\epsilon \in (0, 1/3)$ 
such that both of the values $1/2 - \epsilon$ and $1/2 + \epsilon$
are in the interior of a subinterval of the partition. 
This assertion alone prevents an oblivious configuration based scheduling algorithm 
to schedule jobs of size $1/2 - \epsilon$ and $1/2 + \epsilon$ in the same server 
at the same time, even though they can fit together perfectly in it.  

To complete the proof, it suffices to consider a single server of capacity one 
and assume that jobs have one of the two resource requirements, $1/2 - \epsilon$ 
and $1/2 + \epsilon$, with equal probability. We will now analyze the case in
which the two values are in the interior of different subintervals, as the case in 
which they fall in the same one is clearly worse.

In what follows, for compactness, we define all the vectors to be $2$-dimensional 
with each dimension corresponding to a type, although the 
number of subintervals can be much larger. In other words, we omit the entities 
of the vector that correspond to subintervals
with zero arrivals. Thus, the 
arrival rate vector is given by $\lambda(1/2, 1/2)$. Under an oblivious algorithm, 
the possible maximal feasible configurations 
(configurations that cannot be increased and still be feasible) are
$(2, 0)$ and $(0, 1)$. In particular, configuration $(2, 0)$ is feasible in a best 
case scenario where jobs of size $1/2 - \epsilon$ are mapped to a subinterval 
with the end-bound in $(1/2 - \epsilon, 1/2]$.

It is on the other hand obvious that the configuration $(1, 1)$ is also feasible 
for the job types considered in this example. Hence a workload $\rho=\lambda/\mu$ 
should be feasible if $\mu(1, 1) > \lambda(1/2, 1/2)$ or 
$\rho = \frac{\lambda}{\mu} < 2$. So  $\rho^\star = 2$.
However under the partition assumption, the following conditions should 
hold for any feasible $\rho$
\begin{equation}
\begin{aligned}
&p_1 \mu (2, 0) + p_2 \mu (0, 1) \ge \lambda(1/2, 1/2), \\
&p_1 + p_2 = 1, \quad p_1, p_2 \ge 0.
\end{aligned}
\end{equation}
The maximum $\rho$ is obtained in this case by choosing $p_1=1/3$
and $p_2=1/3$. That is equivalent with $\rho \le 4/3 = 2/3 \rho^\star$.

\subsection{Proof of Corollary~\ref{cor:VQS_J}}\label{pr:VQS_J}
We consider the following $4$ systems which differ in the way that they process jobs of size less than $1/2^J$:
\begin{enumerate}[leftmargin=*]
	\item The jobs are completely discarded from queue and are not 
	processed further
	\item Jobs join the queue without any changes
	\item Jobs join the queue and have their resource requirement rounded to $1/2^J$
	\item Jobs join the queue and have their resource requirement re-sampled from the  
	distribution $F_R$ until their resource value becomes more than $1/2^J$.
\end{enumerate}

We denote the maximum workload achieved in each of the $4$ systems by
$\rho^\star_1, \rho^\star_2, \rho^\star_3, \rho^\star_4$. The 
relation between the job sizes in the systems is increasing. Also the 
distribution of job sizes in the first and last system is the same, 
but in the latter the arrival rate of the jobs is increased by a factor of 
$1/(1- \epsilon)$. It follows that the following relationship must hold between 
the optimal workload of these $4$ systems:
\begin{equation}\label{eq:multiway}
\rho^\star_1 \ge \rho^\star_2 = \rho^\star \ge \rho^\star_3 \ge \rho^\star_4 \ge
\rho^\star_1(1 -\epsilon).
\end{equation} 
The bound of VQS algorithm from Theorem~\ref{thm:VQS1} is valid for the third system, so let $\rho^\star_{VQS}$ be the maximum supportable workload by $VQS$. It then follows from that theorem and inequality \dref{eq:multiway} that
\begin{equation}
\rho^\star_{VQS} \ge \frac{2}{3} \rho^\star_3 \ge \frac{2}{3} \rho^\star_4 =
\frac{2}{3} (1 -\epsilon) \rho^\star_1 \ge \frac{2}{3} (1 -\epsilon) \rho^\star_2
= \frac{2}{3} (1 -\epsilon) \rho^\star
\end{equation}

\subsection{Proof of Theorem~\ref{thm:VQS2}}\label{pr:thm_VQS2}
To prove the throughput result for \algB-\algA, the fundamental
change compared to the proof of Theorem~\ref{thm:VQS1}, is that the proof of
Lemma~\ref{lem:E} needs to make use of new assumptions.
We restate the Lemma next and prove it under the assumptions of the algorithm 
\algB{}-\algA{}.

\begin{lemma*}\label{lem:E2}
	We can ensure $\mathds{P}(E_{{S}(t_0),N_1,N_2,\gamma}) > 1-\epsilon$
	when scheduling under \algB-\algA, if 
	$N_1 >\frac{\log{(\epsilon/(2L))}}{\log{\left(1 - \mu^{K_{max}}\right)}}$
	and $\left\Vert \mathbf{Q}(t_0) \right\Vert > B_{\gamma} \frac{N_2}{\epsilon}$
	for some constant $B_{\gamma}$.
\end{lemma*}
\begin{subproof}
	The first part can be proven under the assumption that once the configuration
	of a server becomes less than $\gamma$ of max-weight configuration, 
	it will become empty in at most $N_1$ time slots for an appropriate value of 
	$N_1$. The analysis of this part is the same as the one in Lemma~\ref{lem:E}
	
	The other condition we need to justify is that a server that becomes active, 
	will schedule according to a configuration that has weight at least $\gamma$ 
	times the optimal, for at least $N_2$ time slots, unless it gets empty again.
	
	We need to distinguish 2 cases for that depending on whether the active 
	configuration of server has a job from $\mVQ_1$ or not. In what follows 
	we highlight only the changes compared to proof of Lemma~\ref{lem:E}.
	
	\textbf{No job from $\mVQ_1$:}
	In this case the server will have $k_{j^\star}$ jobs of type-$j^\star$ 
	in its active configuration for some $j^\star \in [0, 2J-1]$. 
	Given that the jobs 
	in $\mVQ_{j^\star}$ are scheduled from largest to smallest, 
	then the jobs in server will be a superset of those in configuration if 
	$Q_{j^\star}(t_0) > K_{max} N_2$ or 
	$k_{j^\star} Q_{j^\star}(t_0) > K_{max}^2 N_2$.
	Since 
	\begin{equation}
	k_{j^\star} Q_{j^\star}(t_0) \ge \frac{1}{2J}
	\sum_{j=0}^{2J-1} k_{j} Q_{j}(t_0) \ge 
	\frac{\left\Vert \mathbf{Q}(t_0) \right\Vert}{2J}
	\end{equation}
	a sufficient condition can be
	\begin{equation}
	\left\Vert \mathbf{Q}(t_0) \right\Vert > 2J K_{max}^2 N_2
	\end{equation}
	Given this condition, the weight of scheduled configuration 
	in the next $N_2$ time slots will be at least the weight of the 
	active configuration. 
	As a next step we need the weight of active configuration to be at least
	$\gamma$ times the maximum weight for the following $N_2$ time slots, 
	so later arguments are the same as in proof of Lemma~\ref{lem:E}.
	
	\textbf{One job from $\mVQ_1$:}
	Under this condition we will further distinguish two cases depending on the 
	length of the other VQ in configuration which we will assume to be 
	$\mVQ_{j^\star}$.
	Let $U = Q_{j^\star}(t_0) k_{j^\star} + Q_{1}(t_0)$ be
	the weight of the max weight configuration.
	
	\begin{enumerate}
		\item $(\gamma+1)U/2 > Q_{1} \ge U/2$:
		That implies $Q_{j^\star}(t_0) k_{j^\star} > (1-\gamma)U/2$.
		A sufficient condition for this and previous condition to happen is,
		following the procedure for the case of "no jobs from $\mVQ_1$" is
		\begin{equation}
		\left\Vert \mathbf{Q}(t_0) \right\Vert > J (1-\gamma) U.
		\end{equation}
		Then the weight of scheduled configuration 
		will be at least the weight of the active configuration.
		As a next step we need the weight of active configuration to be at least
		$\gamma$ times the maximum weight for the following $N_2$ time slots, 
		so later arguments are the same as in proof of Lemma~\ref{lem:E}.

		\item $Q_{1} \ge (\gamma+1)U/2$:
		In this case we can at least ensure that if 
		\begin{equation}
		\left\Vert \mathbf{Q}(t_0) \right\Vert > J (\gamma+1) U
		\end{equation}
		the $\mVQ_1$ will never empty, but at the same time we need to consider
		the weight of server's configuration assuming that only job of type-$1$ will 
		be in it at all times. For this we consider our configuration has only 
		one job of type-$1$ for which we can claim as opposed to 
		equation~(\ref{eq:emptyA1}) that
		\begin{equation}
		\langle \mathbf{k}(t_{e_\ell(i)}), \mathbf{Q}(t_{e_\ell(i)}) \rangle \ge
		(1+\gamma)/2 \langle \mathbf{k}, \mathbf{Q}(t_{e_\ell(i)}) \rangle.
		\end{equation}
		This leads to the following equivalent of equation~(\ref{eq:DQ})
		
		\begin{equation}
		\begin{aligned}
		&\left \Vert \mathbf{Q}(t_{e_\ell(i)}) - \mathbf{Q}[n]\right \Vert > \\
		&\left \Vert \mathbf{Q}(t_{e_\ell(i)}) \right \Vert
		\frac{\left | \langle \mathbf{k}(t_{e_\ell(i)}) - \frac{1+\gamma}{2}\mathbf{k},
			\mathbf{k}(t_{e_\ell(i)}) - \gamma \mathbf{k} \rangle \right |}
		{\left \Vert \mathbf{k}(t_{e_\ell(i)}) - \frac{1+\gamma}{2}\mathbf{k} \right \Vert
			\left \Vert \mathbf{k}(t_{e_\ell(i)}) - \gamma \mathbf{k} \right \Vert}
		\end{aligned}
		\end{equation}
		with the equivalent of equation~\dref{eq:Cb} being
		\begin{equation}
		B_{\gamma1} = \min_{\mathbf{k^\prime}, \mathbf{k}}
		\frac{\left | \langle \mathbf{k^\prime} - \frac{1+\gamma}{2}\mathbf{k},
			\mathbf{k^\prime} - \gamma \mathbf{k} \rangle \right |}
		{\left \Vert \mathbf{k^\prime} - \frac{1+\gamma}{2}\mathbf{k} \right \Vert
			\left \Vert \mathbf{k^\prime} - \gamma \mathbf{k} \right \Vert}
		\end{equation}
		Later analysis is the same as in proof of Lemma~\ref{lem:E}
		with only the constant $B_{\gamma1}$ being different.
	\end{enumerate}
\end{subproof}

\end{document}